\documentclass[11pt]{article}
\usepackage[utf8]{inputenc}
\usepackage[russian,english]{babel}
\newcommand{\rus}[1]{\foreignlanguage{russian}{#1}}
\usepackage{geometry}             
\usepackage{amssymb,amsmath,amsthm,hyperref,xcolor}

\usepackage{url}

\newtheorem{theorem}{Theorem}
\newtheorem{proposition}{Proposition}
\newtheorem*{corollary}{Corollary}
\newtheorem{lemma}{Lemma}[theorem]

\theoremstyle{definition}
\newtheorem{definition}{Definition}
\theoremstyle{remark}
\newtheorem*{remark}{Remark}
\newtheorem*{question}{Question}
\let\eps=\varepsilon
\let\le=\leqslant
\let\ge=\geqslant
\newcommand{\cnd}{\mskip 1mu | \mskip 1mu}
\newcommand{\dimFS}{\mathrm{dim}_{\mathrm{FS}}}
\DeclareMathOperator{\KS}{\mathrm{C}}
\DeclareMathOperator{\CSR}{\mathrm{CSR}}
\DeclareMathOperator{\KP}{\mathrm{K}}
\DeclareMathOperator{\KA}{\mathrm{KA}}
\DeclareMathOperator{\poly}{\mathrm{poly}}

\newcommand{\nb}[1]{{\color{blue}}}

\title{Automatic Kolmogorov complexity, normality,\\ and finite state dimension revisited}
\date{\today}

\author{Alexander Kozachinskiy\thanks{HSE Computer Science Department, Moscow}\and Alexander Shen\thanks{LIRMM, CNRS and University of Montpellier, France. On leave from IITP RAS, Moscow. Supported by ANR-15-CE40-0016 RaCAF grant. Part of the work was done while visiting National Research University High School of Economics, Moscow. E-mail: \texttt{sasha.shen@gmail.com} or \texttt{alexander.shen@lirmm.fr}}}

\begin{document}
\maketitle

\begin{abstract}
It is well known that normality (all factors of a given length appear in an infinite sequence with the same frequency) can be described as incompressibility via finite automata. Still the statement and the proof of this result as given by Becher and Heiber~\cite{becher-heiber} in terms of ``lossless finite-state compressors'' do not follow the standard scheme of Kolmogorov complexity definition (an automaton is used for compression, not decompression). We modify this approach to make it more similar to the traditional Kolmogorov complexity theory (and simpler) by explicitly defining the notion of automatic Kolmogorov complexity and using its simple properties. Other known notions (Shallit--Wang~\cite{shallit-wang}, Calude -- Salomaa -- Roblot~\cite{csr}) of description complexity related to finite automata are discussed (see the last section).

Using this characterization, we provide easy proofs for most of the classical results about normal sequences, including the equivalence between aligned and non-aligned definitions of normality, the Piatetski-Shapiro  sufficient condition for normality (in a strong form), and Wall's theorem saying that a normal real number remains normal when multiplied by a rational number or when a rational number is added. Using the same characterization, we prove a sufficient condition for normality of a sequence in terms of Kolmogorov complexity. This condition implies the normality of Champernowne's sequence as well as some generalizations of this result (provided by Champernowne himself, Besicovitch, Copeland and Erd\"os). It can be also used to give a simple proof of the result of Calude -- Staiger -- Stephan~\cite{css} saying that normality cannot be characterized in terms of the automatic complexity notion introduced by Calude -- Salomaa -- Roblot~\cite{csr}.

Then we extend this approach to finite state dimension showing that automatic Kolmogorov complexity can be used to characterize the finite state dimension (defined by Dai, Lathrop, Lutz and Mayordomo in~\cite{dai-fsd}). We start with the block entropy definition of the finite state dimension given by Bourke, Hitchcock and Vinogradchandran~\cite{bourke} and show that one may use non-aligned blocks in this definition. Then we show that this definition is equivalent to the definition in terms of automatic complexity. Finally, we use a slightly different version of automatic complexity (a finite state version of a priori  complexity) to show the equivalence between the block entropy definition and original definition from~\cite{dai-fsd} (this equivalence was proven in~\cite{bourke}). We also give a ``machine-independent'' characterization of finite state dimension in terms of superadditive functions that are ``calibrated'' in some sense (have not too many small values), or superadditive upper bounds for Kolmogorov complexity.

Finally, we use our tools to give a simple proof of Agafonov's result saying that normality is preserved by automatic selection rules~\cite{agafonov} as well as the results of Schnorr and Stimm~\cite{schnorr-stimm} that relate normality to finite state martingales. 

Some results of this paper were presented at the Fundamentals in Computing Theory conferences in 2017 and 2019~\cite{fct2017,fct2019}. Preliminary version of this paper (that does not mention the finite state dimension) was published in \texttt{arxiv.org} in~2017~\cite{arxiv2017}.

\end{abstract}

\clearpage
\tableofcontents
\clearpage

\section{Introduction}

What is an individual random object? When could we believe, looking at an infinite sequence $\alpha$ of zeros and ones,  that $\alpha$ was obtained by tossing a fair coin? The minimal requirement is that zeros and ones appear ``equally often'' in $\alpha$: both have limit frequency $1/2$. Moreover, it is natural to require that all $2^k$ bit blocks of length $k$ appear equally often. Sequences that have this property are called \emph{normal} (see the exact definition in Section~\ref{subsec:normal}; a historic account can be found in~\cite{becher-heiber,bugeaud}).

Intuitively, a reasonable definition of an individual random sequence should require much more than just normality; the corresponding notions are studied in the algorithmic randomness theory (see~\cite{downey-hirschfeldt,nies} for the detailed exposition, \cite{kolmbook} for a textbook and \cite{shen-intro} for a short survey). A widely accepted definition of randomness is the given by Martin-L\"of; the corresponding notion is called now \emph{Martin-L\"of randomness}. The classical Schnorr -- Levin theorem says that this notion is equivalent to \emph{incompressibility}: a sequence $\alpha$ is Martin-L\"of random if an only if prefixes of $\alpha$ are incompressible (do not have short descriptions). See again~\cite{downey-hirschfeldt,nies,kolmbook,shen-intro} for exact definitions and proofs.

It is natural to expect that normality, being a weak randomness property, corresponds to some weak incompressibility property. The connection between normality and finite-state computations was noticed long ago, as the title of~\cite{agafonov} shows. This connection led to a characterization of normality as ``finite-state incompressibility'' (see~\cite{becher-heiber} and then~\cite{bch2,carton-talk}). However, the notion of incompressibility that was used in~\cite{becher-heiber} does not fit well the general framework of Kolmogorov complexity (finite automata are considered there as \emph{compressors}, while in the usual definition of Kolmogorov complexity we restrict the class of allowed \emph{decompressors}).

In this paper we give a definition of automatic Kolmogorov complexity that restricts the class of allowed decompressors and is suitable for the characterization of normal sequences as incompressible ones. This definition and its properties are considered in Section~\ref{sec:complexity}. Section~\ref{sec:normal-incompressible} presents one of our main results: characterization of normality in terms of automatic complexity.  First (Section~\ref{subsec:normal}) we recall the notion of a normal sequence. Then (Section~\ref{subsec:charac}) we provide a characterization of normal sequences as sequences whose prefixes have automatic Kolmogorov complexity close to the length.

This characterization is used in Section~\ref{sec:applications} to provide simple proofs for many classical results about normal sequences. In Section~\ref{subsec:non-aligned} we give a simple proof of an old result (Borel, Pillai, Niven -- Zuckerman,~\cite{borel1909,pillai1940,niven-zuckerman1951}) saying that normality can be equivalently defined in terms of frequencies of aligned or non-aligned blocks (we get the same notion in both cases). In Section~\ref{subsec:piatetski-shapiro} we provide a simple proof of the result by Piatetski-Shapiro~\cite{piatetski-shapiro,piatetski-shapiro1957} saying that a sequence is normal if for every $k$-bit block its frequency is not much bigger than its expected frequency in a random sequence. This proof can be used to prove a stronger version of this result, replacing a constant factor in Piatetski-Shapiro version by factor $2^{o(k)}$. We note also that Piatetski-Shapiro's result easily implies Wall's theorem (saying that normal numbers remain normal when multiplied by a rational factor).
 
Then in Section~\ref{subsec:champernowne} we return to the first example of a normal sequence given by Champernowne~\cite{champernowne} and show that its normality easily follows from a simple sufficient condition for normality in terms of Kolmogorov complexity (Theorem~\ref{th:normality-condition}) . The same sufficient condition easily implies the generalizations of Champernowne's results obtained by Copeland--Erd\"os (\cite{copeland-erdos}, see Section~\ref{subsec:copeland-erdos}) and Besicovitch (\cite{besicovitch}, see Section~\ref{subsec:besicovitch}).  Finally, in Section~\ref{subsec:csr} we show how this sufficient condition gives a simple proof of a result proven by Calude, Staiger and Stephan~\cite{css}. It says that the definition of automatic complexity from~\cite{csr} does not provide a criterion of normality (this question was asked in~\cite{csr}).

The notion of normality can be interpreted as ``weak randomness'' (weak incompressibility). Instead of randomness, one can consider a more general notion of \emph{effective Hausdorff dimension} introduced by Lutz in~\cite{lutz-dimension} (see~\cite[Sections 5.8 and 9.10]{kolmbook} for details). The effective Hausdorff dimension is defined for arbitrary binary sequences and is between $0$ and $1$; the smaller it is, the more compressible is the sequence. Formally speaking, the effective Hausdorff dimension of a sequence $\alpha=a_0a_1\ldots$ can be defined in terms of Kolmogorov complexity as $\liminf_n [(\text{complexity of $a_0\ldots a_{n-1}$})/n$]. For random sequences the effective Hausdorff dimension equals $1$ (as well as for some non-random sequences, e.g., for a sequence that is obtained from a random one by replacing all terms $a_{2^n}$ by zeros). This notion is an effective counterpart of the classical notion of Hausdorff dimension, see~\cite{lutz-dimension,kolmbook}.

The notion of effective Hausdorff dimension has a scaled-down version  based on  finite automata. This notion is called \emph{finite-state dimension} and was introduced in~\cite{dai-fsd}. In this paper it was defined in terms of finite-state martingales; in~\cite{bourke} an equivalent definition in terms of entropy rates was provided.  In Section~\ref{subsec:entropy-rates} we revisit the definition of finite-state dimension in terms of entropy rates and show that one may use both aligned and non-aligned blocks in this definition and get the same quantity. However, this equivalence does not work for blocks of fixed size, as the counterexamples of Section~\ref{subsec:large-blocks} (Theorem~\ref{thm:large-blocks}) show. In Section~\ref{subsec:fsd-wall} we give a simplified proof of a theorem of Doty, Lutz and Nandakumar~\cite{doty} saying that finite-state dimension does not change when a real number is multiplied by a rational factor. Then in Section~\ref{subsec:fsd-ac} we show that finite-state dimension can be characterized in terms of automatic complexity as the $\liminf$ of complexity/length ratio for prefixes, thus giving a characterization of finite-state dimension that is parallel to the characterization of effective Hausdorff dimension in terms of Kolmogorov complexity. Again, this connection between finite-state dimension and automatic compression was noted long ago in~\cite{dai-fsd}. Our goal here is to give a statement that is parallel to the Kolmogorov complexity characterization of effective Hausdorff dimension. The only difference is that we use automatic Kolmogorov complexity instead of the standard one (and have to take infimum over all automata since there is no universal automaton that leads to minimal complexity). To prove this characterization, we use the definition of finite-state dimension in terms of entropy rates.   In Section~\ref{subsec:stateless}  we show that this characterization is quite robust: automatic complexity can be replaced by other similar notions. For example, we may consider all superadditive upper bounds for Kolmogorov complexity (Theorem~\ref{th:stateless}), or even give a characterization of finite-state dimension (Theorem~\ref{th:stateless1}) that does not mention entropy, Kolmogorov complexity and finite-state machines at all, and just considers a class of superadditive functions that are ``calibrated'' in some natural sense (have not too many small values).

In Section~\ref{subsec:fsm} we give a simple proof that the definition of finite-state dimension in terms of entropy rates is equivalent to the original definition from~\cite{dai-fsd}. For that we discuss a finite-state version of the notion of a priori probability (maximal semimeasure) used in the algorithmic information theory, and show that it also can be used to characterize the finite-state dimension. In Section~\ref{subsec:agafonov} we use our tools to give a simple proof for the result of Agafonov~\cite{agafonov} (finite automaton selects a normal sequence if applied to a normal sequence) and its extension by Schnorr and Stimm~\cite{schnorr-stimm} saying the any finite-state martingale is either constant or exponentially decreases on sufficiently long prefixes of a normal sequence. We also mention a natural notion of finite-state measure that generalizes the notion of multi-account gales~\cite{dai-fsd} and also can be used to characterize finite-state dimension (Section~\ref{subsec:multi-account}). 

The notions of Hausdorff dimension, effective Hausdorff dimension and finite-state dimension have \emph{strong} counterparts~\cite{athreya-strong}, known as packing dimension, effective strong dimension and finite-state strong dimension. They can be defined in terms of martingales by requiring that winning martingale is large not only infinitely often (as for the Hausdorff dimension) but for all sufficiently long prefixes. In terms of complexities, we consider $\limsup$ of complexity/length ratios instead of $\liminf$. We note (Section~\ref{subsec:strong}) that all the results relating finite-state dimension, automatic complexity and automatic a priori probability remain valid (with almost the same proofs) for the strong dimensions. This include one of the oldest results of this type relating automatic compression rate with the limit entropy of (non-aligned) blocks that goes back to Lempel and Ziv~\cite{ziv-lempel-variable-rate}, see also their earlier papers~\cite{lempel-ziv-complexity,ziv-lempel-universal,ziv}.

Finally, in Section~\ref{sec:discussion} we compare our definition of automatic complexity with other similar notions.

\section{Automatic Kolmogorov complexity}\label{sec:complexity}

\subsection{General scheme of defining complexities}\label{subsec:general-scheme}

The algorithmic (Kolmogorov) complexity is usually defined in the following way: $\KS(x)$, the complexity of an object $x$, is the minimal length of its ``description''. We assume that both objects and descriptions are binary strings; the set of binary strings is denoted by $\mathbb{B}^*$, where $\mathbb{B}=\{0,1\}$. Of course, this definition makes sense only after we explain which type of ``descriptions'' we consider, but most versions of Kolmogorov complexity can be defined according to this scheme~\cite{uspshen}.

\begin{definition}
Let $D\subset \mathbb{B}^*\times\mathbb{B}^*$ be a binary relation; we read $(p,x)\in D$ as ``$p$ is a $D$-description of $x$''. Then \emph{complexity function} $\KS_D$ is defined as
$$
\KS_D(x)=\min\{ |p|\colon (p,x)\in D\},
$$
i.e., as the minimal length of a $D$-description of $x$.
\end{definition}
Here $|p|$ stands for the length of a binary string $p$ and $\min(\varnothing)=+\infty$, as usual. We say that $D$ is a \emph{description mode} and $\KS_D(x)$ is the \emph{complexity of $x$ with respect to the description mode~$D$}.

We get the original version of Kolmogorov complexity (``plain complexity'') if we consider all computable partial functions as description modes, i.e., if we consider relations $D_f=\{(p,f(p))\}$ for arbitrary computable partial functions $f$ as description modes. Equivalently, we may say that we consider (computably) enumerable relations $D$ that are graphs of functions (for every $p$ there exists at most one $x$ such that $(p,x)\in D$; each description describes at most one object). Then the Kolmogorov -- Solomonoff optimality theorem says that there exists an optimal $D$ in this class that makes $\KS_D$ minimal up to an $O(1)$ additive term. We assume that the reader is familiar with basic properties of Kolmogorov complexity, see, e.g., ~\cite{lv,kolmbook}; for a short introduction see also~\cite{shen-intro}.

Note that we could get a trivial $\KS_D$ if we take, e.g., the set of all pairs as a description mode $D$. In this case all strings have complexity zero, since the empty string describes all of them. So we should be careful and do not consider description modes where the same string describes too many different objects.

\subsection{Automatic description modes}\label{subsec:automatic-modes}

To define our class of description modes, let us first recall some basic notions related to finite automata. Let $A$ and $B$ be two finite alphabets. Consider a directed graph $G$ whose edges are labeled by pairs $(a,b)$ of letters (from $A$ and $B$ respectively). We also allow pairs of the form $(a,\eps)$, $(\eps,b)$, and $(\eps,\eps)$ where $\eps$ is a special symbol (not in $A$ or $B$) that informally means ``no letter''. For such a graph, consider all directed paths in it (no restriction on starting or final points), and for each path $p$ concatenate all the first components and also (separately) all the second components of the pairs along the path; $\eps$ is replaced by an empty word. For each path we get some pair $(u_p,v_p)$ where $u_p\in A^*$ and $v_p\in B^*$ (i.e., $u_p$ and $v_p$ are words over alphabets $A$ and $B$). Consider all pairs that can be read in this way along all paths in $G$, i.e., consider the set $R_G=\{(u_p,v_p)\mid \text{$p$ is a path in $G$}\}$. For each labeled graph $G$ we obtain a relation $R_G$ that is a subset of $A^*\times B^*$. For the purposes of this paper, we call the relations obtained in this way ``automatic''.  This notion is similar to \emph{rational relations} defined by transducers~\cite[Section III.6]{berstel}. The difference is that we do not fix initial/finite states (so every sub-path of a valid path is also valid) and that we do not allow arbitrary words as labels, only letters and $\eps$. (This will be important, e.g., for the statement (j) of Theorem~\ref{th:complexity}.)

\begin{definition}\label{def:automatic}
A relation $R\subset A^*\times B^*$ is \emph{automatic} if there exists a labeled graph \textup(\emph{automaton}\textup) $G$ such that $R=R_G$.
\end{definition}

Now we define automatic description modes as automatic relations where each string describes at most $O(1)$ objects.

\begin{definition}
A relation $D\subset \mathbb{B}^*\times\mathbb{B}^*$ is an \emph{automatic description mode} if 
\begin{itemize}
\item $D$ is automatic in the sense of Definition~\textup{\ref{def:automatic};}
\item $D$ is a graph of an $O(1)$-valued function: there exists some constant $c$ such that for each $p$ there are at most $c$ values of $x$ such that $(p,x)\in D$.
\end{itemize}
\end{definition}

For every automatic description mode $D$ we consider the corresponding complexity function~$\KS_D$. There is no optimal mode $D$ that makes $\KS_D$ minimal (see Theorem~\ref{th:complexity} below). So, stating some properties of complexity, we need to mention $D$ explicitly. Moreover, for a statement that compares the complexities of different strings, we need to say something like ``for every automatic description mode $D$ there exists another automatic description mode $D'$ such that\ldots'', and then make a statement that involves both $\KS_D$ and $\KS_{D'}$. (A similar approach is needed when we try to adapt inequalities for Kolmogorov complexity to the case of resource-bounded complexities.)

\subsection{Properties of automatic description modes}\label{subsec:properties-modes}

Let us first mention some basic properties of automatic description modes.

\begin{proposition}\label{prop:modes}
\leavevmode
\begin{description}
\item[\textup{(a)}] The union of two automatic description modes is an automatic description mode.
\item[\textup{(b)}] The composition of two automatic description modes is an automatic description mode.
\item[\textup{(c)}] If $D$ is a description mode, then $\{(p,x0)\colon (p,x)\in D\}$ is a description mode \textup(here $x0$ is the binary string $x$ with $0$ appended\textup); the same is true for $x1$ instead of $x0$.
\end{description}
\end{proposition}

\begin{proof}
There are two requirements for an automatic description mode:  (1)~the relation is automatic and (2)~the number of images is bounded. The second one is obvious in all three cases. The first one can be proven by a standard argument (see, e.g., \cite[Theorem 4.4]{berstel}) that we reproduce for completeness.

(a)~The union of two relations $R_G$ and $R_G'$ for two automata $G$ and $G'$ corresponds to an automaton that is a disjoint union of $G$ and $G'$.

(b)~Let $S$ and $T$ be automatic relations that correspond to automata $K$ and~$L$. Consider a new graph that has set of vertices $K\times L$. (Here we denote an automaton and the set of vertices of its underlying graph by the same letter.)
\begin{itemize}
\item If an edge $k\to k'$ with a label $(a,\eps)$ exists in $K$, then the new graph has edges $(k,l)\to(k',l)$ for all $l\in L$; all these edges have the same label $(a,\eps)$. 
\item In the same way an edge $l\to l'$ with a label $(\eps,c)$ in $L$ causes edges $(k,l)\to (k,l')$ in the new graph for all $k$; all these edges have the same label~$(\eps,c)$.
\item Finally, if $K$ has an edge $k\to k'$ labeled $(a,b)$ and at the same time $L$ has an edge $l\to l'$ labeled $(b,c)$, where $b$ is the same letter, then we add an edge $(k,l)\to(k',l')$ labeled $(a,c)$ in the new graph.
\end{itemize}
Any path in the new graph is projected into two paths in $K$ and $L$. Let $(p,q)$ and $(u,v)$ be the pairs of words that can be read along these projected paths in $K$ and $L$ respectively, so $(p,q)\in S$ and $(u,v)\in T$. The  construction of the graph $K\times L$ guarantees that $q=u$ and that we read $(p,v)$ in the new graph along the path. So every pair $(p,v)$ of strings that can be read in the new graph belongs to the composition of $S$ and $T$.

On the other hand, assume that $(p,v)$ belong to the composition, i.e., there exists $q$ such that $(p,q)$ can be read along some path in $K$ and $(q,v)$ can be read along some path in $L$. Then the same word $q$ appears in the second components in the first path and in the first components in the second path. If we align the two paths in such a way that the letters of $q$ appear at the same time, we get a valid transition of the third type for each letter of $q$. Then we complete the path by adding transitions in between the synchronized ones (interleaving them in arbitrary way); all these transitions exist in the new graph by construction.

(c)~We add an additional outgoing edge labeled $(\eps,0)$ for each vertex of the graph; all these edges go to a special vertex that has no outgoing edges.
\end{proof}

\begin{remark}
Given a graph, one can check in polynomial time whether the corresponding relation is $O(1)$-valued~\cite[Theorem 5.3, p.~777]{weber}.
\end{remark}

\subsection{Properties of automatic complexity}\label{subsec:properties-automatic-complexity}

Now we are ready to prove the following simple result about the properties of \emph{automatic Kolmogorov complexity} functions, i.e., of functions $\KS_R$ where $R$ is some automatic description mode.

\begin{theorem}[Basic properties of automatic Kolmogorov complexity]\label{th:complexity}
\leavevmode
\begin{description}
\item[\textup{(a)}] There exists an automatic description mode $R$ such that $\KS_R(x)\le |x|$ for all strings $x$.
\item[\textup{(b)}] For every automatic description mode $R$ there exists some automatic description mode $R'$ such that  $\KS_{R'}(x0)\le \KS_R(x)$ and $\KS_{R'}(x1)\le \KS_R(x)$ for all~$x$.
\item[\textup{(c)}] For every automatic description mode $R$ there exists some automatic description mode $R'$ such that $\KS_{R'}(\bar x)\le \KS_R(x)$, where $\bar x$ stands for the reversed~$x$. 
\item[\textup{(d)}] For every automatic description mode $R$ there exists some constant $c$ such that $\KS(x)\le \KS_R(x)+c$. \textup(Here $\KS$ stands for the plain Kolmogorov complexity.\textup)
\item[\textup{(e)}] For every automatic description mode $R$ there exists some constant $c$ such that for every $n$ there exist at most $c2^n$ strings $x$ such that $\KS_R(x)<n$.
\item[\textup{(f)}] For every $c>0$ there exists an automatic description mode $R$ such that\\ $\KS_R(1^n)\le n/c$ for all n. 
\item[\textup{(g)}] For every automatic description mode $R$ there exists some $c>0$ such that $\KS_R(1^n)\ge n/c-1$ for all $n$.
\item[\textup{(h)}] For every two automatic description modes $R_1$ and $R_2$ there exists an automatic description mode $R$ such that $\KS_{R}(x)\le \KS_{R_1}(x)$ and $\KS_R(x)\le \KS_{R_2}(x)$ for all $x$.
\item[\textup{(i)}] There is no optimal mode in the class of automatic description modes. \textup(A mode $R$ is called optimal in some class if for every mode $R'$ in this class there exists some $c$ such that $\KS_R(x)\le\KS_{R'}(x)+c$ for all strings $x$.\textup)
\item[\textup{(j)}] For every automatic description mode $R$, if $x'$ is a substring of $x$, then $\KS_R(x')\le \KS_R(x)$.
\item[\textup{(k)}] Moreover, $\KS_R(xy)\ge \KS_R(x)+\KS_R(y)$ for every two strings $x$ and $y$.
\item[\textup{(l)}] For every automatic description mode $R$ and for every constant $\eps>0$ there exists an automatic description mode $R'$ such that $\KS_{R'}(xy)\le (1+\eps)\KS_R(x)+\KS_R(y)$ for all strings~$x$ and~$y$. 
\item[\textup{(m)}] Let $S$ be an automatic description mode. Then for every automatic description mode $R$ there exists an automatic description mode $R'$ such that $\KS_{R'}(y)\le \KS_{R}(x)$ for every $(x,y)\in S$.
\item[\textup{(n)}] If we allow a bigger alphabet $B$ instead of $\mathbb{B}=\{0,1\}$ as an alphabet for descriptions, then the complexity becomes $\log|B|$ times smaller, up to a constant factor that can be chosen arbitrarily close to~$1$. More precisely, for every automatic description mode $D$ with arbitrary alphabet $B$ and every $\eps>0$ there exist an automatic description mode $D'$ with binary alphabet such that
$$
\KS_{D'}(x)\le (1+\eps)\log |B|\KS_D(x)
$$
for all sufficiently long $x$.
\end{description}
\end{theorem}

\begin{proof}\leavevmode
(a)~Consider an identity relation as a description mode; it corresponds to an automaton with one state.

(b)~This is a direct corollary of Proposition~\ref{prop:modes},~(c).

(c)~The definition of an automaton is symmetric (all edges can be reversed), and the $O(1)$-condition still holds.

(d)~Let $R$ be an automatic description mode. An automaton defines a decidable (computable) relation, so $R$ is decidable. Since $R$ defines a $O(1)$-valued function, a Kolmogorov description of some $y$ that consists of its $R$-description $x$ and the ordinal number of $y$ among all strings that are in $R$-relation to $x$ (in some natural ordering), is only $O(1)$ bits longer than~$x$.

(e)~This is a direct corollary of $d$, since there are less than $2^n$ strings of Kolmogorov complexity less than $n$. Or we may just count all the descriptions of length less than $n$. There are less than $2^n$ of them, and each describes only $O(1)$ strings.

(f)~Consider an automaton that consists of a cycle where it reads one input symbol $1$ and then produces $c$ output symbols $1$.  Here we consider first components of pairs as ``input symbols'' and second components as ``output symbols''  since the relation is considered as an $O(1)$-multivalued function. Recall that there are no restrictions on initial and finite states, so this automaton produces all pairs $(1^k, 1^l)$ where $(k-1)c\le l \le (k+1)c$.

(g)~Consider an arbitrary description mode, i.e., an automaton that defines some $O(1)$-valued relation. Then every cycle in the automaton that produces some output letter should also produce some input letter, otherwise an empty input string corresponds to infinitely many output strings. For any sufficiently long path in the graph we can cut away a minimal cycle, removing at least one input letter and at most $c$ output letters, where $c$ is the number of states, until we get a path of length less than~$c$.

(h)~This follows from Proposition~\ref{prop:modes},~(a).

(i)~This statement is a direct consequence of~(f) and~(g). Note that for finitely many automatic description modes there is a mode that is better than all of them, as (h) shows, but we cannot do the same for all description modes (as was the case for Kolmogorov complexity).

(j)~If $R$ is a description mode, $(p,x)$ belongs to $R$ and $x'$ is a substring of $x$, then there exists some substring $p'$ of $p$ such that $(p',x')\in R$. Indeed, we may consider the input symbols used while producing~$x'$.

(k)~Note that in the previous argument we can choose disjoint $p'$ for disjoint~$x'$.

(l)~Informally, we modify the description mode as follows: a fixed fraction of input symbols is used to indicate when a description of $x$ ends and a description of $y$ begins. More formally, let $R$ be an automatic description mode; we use the same notation $R$ for the corresponding automaton. Consider $N+1$ copies of $R$ (called $0$-, $1$-,\ldots, $N$-th layers). The outgoing edges from the vertices of $i$-th layer that contain an input symbol are redirected to $(i+1)$-th layer (the new state remains the same, only the layer changes, so the layer number counts the input length). The edges with no input symbol are left unchanged (and go to $i$-th layer as before). The edges from the $N$-th layer are of two types: for each vertex $x$ there is an edge with label $(0,\eps)$ that goes to the same vertex in $0$-th layer, and edges with labels $(1,\eps)$ that connect each vertex of $N$-th layer to all vertices of an additional copy of $R$ (so we have $N+2$ copies in total). If both $x$ and $y$ can be read (as outputs) along the edges of $R$, then $xy$ can be read, too (additional zeros should be added to the input string after groups of $N$ input symbols). We switch from $x$ to $y$ using the edge that goes from $N$th layer to the additional copy of~$R$ (using additional symbol $1$ in the input string). The overhead in the description is one symbol per every $N$ input symbols used to describe $x$. We get the required bound, since $N$ can be arbitrarily large.

The only thing to check is that the new automaton is $O(1)$-valued. Indeed, the possible switch position (when we move to the states of the additional copy of $R$) is determined by the positions of the auxiliary bits modulo $N+1$: when this position modulo $N+1$ is fixed, we look for the first $1$ among the auxiliary bits. This gives only a bounded factor $(N+1)$ for the number of possible outputs that correspond to a given input.

(m)~The composition $S\circ R$ is an automatic description mode due to Proposition~\ref{prop:modes}, (b).
  
(n)~Take the composition of a given description mode $R$ with a mode that provides block  encoding of inputs. Note that block encoding can be implemented by an automaton. There is some overhead when $|B|$ is not a power of $2$, but the corresponding factor becomes arbitrarily close to $1$ if we use block code with large block size.
\end{proof}

Not all these results are used in the sequel; we provide them for comparison with the properties of the standard Kolmogorov complexity function.  Still let us introduce a name for the property (k) since it plays an important role in the sequel.  

\begin{definition}\label{def:superadditive}
A function $f$ with non-negative values defined on binary strings is called \emph{superadditive} if $f(xy)\ge f(x)+f(y)$ for any two strings $x$ and $y$. 
\end{definition}

Now (k) can be reformulated as follows: \emph{for every automatic description mode $R$ the function $\KS_R$ is superadditive}. As we will see (Sections~\ref{subsec:fsd-ac} and~\ref{subsec:stateless}), the characterization of normality (and finite-state dimension, see below) in terms of automatic complexity can be extended to all upper bounds for Kolmogorov complexity that are superadditive, and also to all superadditive functions satisfying the property (e). Note that usual Kolmogorov complexity is not superadditive for obvious reasons: say, $\KS(xx)$ is close to $\KS(x)$, not to $2\KS(x)$. Note also that a superadditive function equals $0$ for the empty string (let $x$ and $y$ be empty in the definition).

\section{Normality and incompressibility}\label{sec:normal-incompressible}

\subsection{Normal sequences and numbers}\label{subsec:normal}

Consider an infinite bit sequence $\alpha=a_0a_1a_2\ldots$ and some integer $k\ge 1$. Split the sequence $\alpha$ into $k$-bit blocks: $\alpha=A_0A_1\ldots$. For every $k$-bit string $r$ consider the limit frequency of $r$ among the $A_i$, i.e. the limit of $\#\{i\colon  i< N \text{ and } A_i=r\}/N$ as $N\to\infty$. This limit may exist or not; if it exists for some $k$ and for all $r$, we get a probability distribution on $k$-bit strings. 

\begin{definition}\label{def:normal}
A sequence $\alpha$ is \emph{normal} if for every number $k$ and every string $r$ of length $k$ this limit exists and is equal to~$2^{-k}$.
\end{definition}

Sometimes sequences with these properties are called \emph{strongly normal} while the name ``normal'' is reserved for sequences that have this property for $k=1$.

There is a version of the definition of normal sequences that considers \emph{all} occurrences of some string $r$ in $\alpha$ (while Definition~\ref{def:normal} considers only \emph{aligned} ones, whose starting point is a multiple of $k$). In this ``non-aligned'' version we require that the limit of $\#\{i<N: \alpha_i \alpha_{i+1}\ldots\alpha_{i+k-1}=r\}/N$ equals $2^{-k}$ for all $k$ and for all strings $r$ of length $k$.  A classical result\footnote{In fact, this result has a rather complicated history. The original definition of normal numbers was given by Borel~\cite{borel1909}. He required that every $k$-bit strings appears with frequency $2^{-k}$ among blocks that we get when we delete some finite prefix of the sequence and cut the rest into $k$-bit blocks. This implies both aligned and non-aligned normality (the aligned normality is the special case when the prefix is empty, the non-aligned normality can be shown by averaging frequencies for prefixes of length $0,1,\ldots,k-1$). Borel noted that his definition follows from non-aligned normality (``La propri\'et\'e caract\'eristique'', p.~261). However, he gave no proof, and the relation between these three definitions (aligned, non-aligned and Borel's definition that implies both) was clarified much later. Pillai~\cite{pillai1940}, correcting his earlier paper~\cite{pillai1939},  showed that aligned normality implies Borel's definition. Niven and Zuckerman~\cite{niven-zuckerman1951} gave a proof of Borel's claim.  Cassels~\cite{cassels1952} provided an alternative proof for the result of Niven and Zuckerman, while Maxfield~\cite{maxfield1952} provided an alternative proof for the result of Pillai. See also~\cite{niven1956}; a more recent exposition can be found, e.g., in a book of Kuipers and Niederreiter~\cite[Chapter 1, Section 8]{niederreiter}; it uses as a tool the Piatetski-Shapiro criterion (see Section~\ref{subsec:piatetski-shapiro} below),  Bugeaud's book~\cite[Sect.~4.1, Equivalent definitions of normality]{bugeaud}, or in the Becher -- Carton chapter in a recent collection~\cite[Theorem 10]{bc}. Even the latter exposition is quite technical and does not use the relation between normality and finite-state machines, though this relation is presented later in the chapter~\cite[sect~5]{bc}.} says that this is an equivalent notion, and we give below (Section~\ref{subsec:non-aligned}) a simple proof of this equivalence using automatic complexity.  Before this proof is given, we will distinguish the two definitions by using the name ``non-aligned-normal'' for the second version.

A real number is called \emph{normal} if its binary expansion is normal (we ignore the integer part). If a number has two binary expansions, like $0.0111\ldots=0.1000\ldots$, both expansions are not normal, so this is not a problem.

A classical example of a normal number is the \emph{Champernowne number}~\cite{champernowne} 
    $$0.0\, 1\, 10\, 11\, 100\, 101\, 110\, 111\, 1000\,1001\ldots$$
(the concatenation of all positive integers in binary). Let us sketch the (standard) proof of its normality using the non-aligned version of normality definition.\footnote{Later we will derive the normality of this sequence from Theorem~\ref{th:normality-condition}. Still we want to give an idea what kind of arguments can be avoided by using our tools.}  All $N$-bit numbers in the Champernowne sequence form a block that starts with $10^{N-1}$ and ends with $1^N$. Note that every string of length $k\ll N$ appears in this block with probability close to $2^{-k}$, since each of $2^{N-1}$ strings (after the leading $1$ for the $N$-bit numbers in the Champernowne sequence) appears exactly once. The deviation is caused by the leading $1$'s and also by the boundaries between the consecutive $N$-bit numbers where the $k$-bit substrings are out of control. Still the deviation is small since $k\ll N$.

This is not enough to conclude that the Champernowne sequence is (non-aligned) normal, since the definition speaks about frequencies in all prefixes; the prefixes that end on a boundary between two blocks are not enough. The problem appears because the size of a block is comparable to the length of the prefix before it. To deal with arbitrary prefixes, let us note that if we ignore \emph{two} leading digits in each number (first $10$ and then $11$) instead of one, the rest is periodic in the block (the block consists of two periods). If we ignore three leading digits, the block consists of four periods, etc. An arbitrary prefix is then close to the boundary between these sub-blocks, and the distance can be made small compared to the total length of the prefix. (End of the proof sketch.)

In fact, the full proof that follows this sketch is quite tedious. There are much more general reasons why this number is normal, as we will see in Section~\ref{subsec:champernowne}, where this result becomes an immediate corollary of the sufficient condition for normality in terms of Kolmogorov complexity (and this condition in its turn is a easy consequence of the criterion of normality in terms of automatic complexity).

\medskip

The definition of normality can be given for an arbitrary alphabet (instead of the binary one), and we get the notion of \emph{$b$-normality} of a real number for every base $b\ge 2$. It is shown by Cassels~\cite{cassels1959} and Schmidt~\cite{schmidt1960} that for different bases we get non-equivalent notions of normal real numbers; the proof is rather difficult. The numbers in $[0,1]$ that are normal for every base are called \emph{absolutely normal}. Their existence can be proved by a probabilistic argument. Indeed, for every base $b$, almost all reals are $b$-normal (the non-normal numbers have Lebesgue measure $0$ by the Strong Law of Large Numbers). Therefore the numbers that are not absolutely normal form a null set (a countable union of the null sets for each $b$). The constructive version of this argument shows that there exist computable absolutely normal numbers. This result goes back to an unpublished note of Turing (1938, see~\cite{becher-turing}).

In the next section we prove the connection between normality and automatic complexity (Theorem~\ref{th:compress}): a sequence $\alpha$ is normal if for every automatic description mode $D$ the complexities $\KS_D$ of its prefixes never become much smaller than their lengths.

\subsection{Normality and incompressibility}\label{subsec:charac}

\medskip
\begin{theorem}\label{th:compress}
A sequence $\alpha=a_0a_1a_2\ldots$ is normal if and only if 
$$
\liminf_{n\to\infty} \frac{\KS_R(a_0a_1\ldots a_{n-1})}{n}\ge1
$$
for every automatic description mode $R$.
\end{theorem}

\begin{proof}
First, let us show that a sequence that is not normal is compressible. Here is the sketch: Assume that for some bit sequence $\alpha$ and for some $k$ the requirement for aligned $k$-bit blocks is not satisfied. Using compactness arguments, we can find a sequence of lengths $N_i$ such that for the prefixes of these lengths the frequencies of $k$-bit blocks do converge to some probability distribution  $P$ on $\mathbb{B}^k$, but this distribution is not uniform.  Then its Shannon entropy $H(P)$ is less than $k$, and there is a prefix-free code for $P$ of small average length. We use this code to get an efficient automatic description for the prefixes of length $N_i$. Let us explain the details.

Recall the basic notions of Shannon information theory (see, e.g.,~\cite[Sect.~7.1]{kolmbook}). Consider a random variable $\pi$ with finite range. It corresponds to a  probability distribution $P$ on its range. The entropy of $\pi$ (or $P$, they can be used interchangeably) is defined as follows. Assume that the range consists of $m$ elements having probabilities $p_1,\ldots,p_m$. Then, by definition,
$$
H(P)=H(\pi)=\sum_{i=1}^m p_i\log \frac{1}{p_i}.
$$
If some $p_i$ are zeros, the corresponding terms are omitted (this is natural since $p\log (1/p)$ converges to $0$ as $p\to 0$). The Shannon entropy of a variable with $m$ values is at most $\log m$, and it is equal to $\log m$ if and only if the distribution is uniform (all $m$ elements have the same probability). 

The Shannon entropy is related to the average length of prefix codes. A \emph{prefix-free code} for a variable with $m$ values is a $m$-tuple of binary strings that is prefix-free (none of the strings is a prefix of another one, so the code can be uniquely decoded from left to right). These strings (called \emph{codewords}) encode $m$ values of the random variable, and the average length of the code is defined as $\sum_{i=1}^m p_i |x_i|$ where $p_i$ is the probability of $i$th value and $x_i$ is its encoding (and $|x_i|$ stands for the length). A basic result of the Shannon information theory guarantees that (a)~the average length of every prefix-free code for $\pi$ is at least $H(P)$, and (b)~there exists a prefix-free code of average length close to $H(P)$, namely, of length at most $H(P)+1$ (Shannon -- Fano code)\footnote{This ``$+1$'' overhead is due to rounding if the frequencies are not powers of $2$. To prove this result, for each $p_i$ we consider the minimal integer $k_i$ such that $2^{-k_i}\le p_i$, and encode $i$th letter by a binary string of length $k_i$.}.  The following lemma uses this code to construct an automatic description mode.

\begin{lemma}\label{lem:shannon-fano}
Let $k$ be some integer and let $P$ be a distribution on a set $\mathbb{B}^k$ of $k$-bit blocks. Then there exists an automatic description mode $R$ such that for every string $x$ whose length is a multiple of $k$, we have
\[
\KS_R(x)\le \frac{|x|}{k}\left(\sum_B Q(B) \log\frac{1}{P(B)} + 1\right)
\]
where $Q$ is the distribution on $k$-bit blocks appearing when $x$ is split into blocks of size $k$.
\end{lemma}

\begin{proof}[Proof of Lemma~\ref{lem:shannon-fano}]
Consider the Shannon -- Fano code for $k$-bit blocks based on the distribution $P$. Then the length of the codeword for arbitrary block $B$ is at most $\log (1/P(B))+1$. This code is prefix-free and can be uniquely decoded bit by bit by a finite automaton that reads the input string until a codeword is found, and then outputs the coressponding block and starts waiting for the next codeword. Therefore, this code corresponds to some automatic description mode $R$.  A string $x$ is a concatenation of $|x|/k$ blocks of length $k$, and has a description whose length is the sum of the lengths of the codes for these blocks.  Each block $B$ has frequency $Q(B)$, i.e., appears $\frac{|x|}{k}Q(B)$ times, and this number should be multiplied by the codeword length, which is at most $\log (1/P(B))+1$.  The overhead term adds $Q(B)$ for each block $B$, so we get $1$ in total.
\end{proof}

Note that this lemma allows some values $P(B)$ to be zeros; if such a block appears in $x$, the right hand side is infinite and the inequality is vacuous. 

We apply this lemma to $N_i$-bit prefixes of $\alpha$ for which the corresponding distributions $Q_i$ on $\mathbb{B}^k$ converge to some distribution $P$ that is not uniform, so $H(P)<k$. We want to construct an automatic description mode $R$ such that $\liminf \KS_R(x_i)/|x_i|<1$ where $x_i$ is the $N_i$-bit prefix of $\alpha$, thus proving Theorem~\ref{th:compress} in one direction.

Assume first that $H(P)<k-1$ and $P$ is everywhere positive (no blocks have zero probability). Then Lemma~\ref{lem:shannon-fano} applied to $x_i$ gives the desired result immediately. Indeed, $Q_i(B)$ converge to $P(B)$ and the sum over $B$ in the right hand side converges to $H(P)$.

It remains to deal with the two problems mentioned. We start with the first one: what to do if $H(P)$ is close to $k$ and the gap is less than~$1$. In this case we switch to larger blocks to get the gap we need. It is done in the following way.

Selecting a subsequence, we may assume without loss of generality that the limit frequencies exist also for (aligned) $2k$-bit blocks, so we get a random variable $P_0P_1$ whose values are $2k$-bit blocks (and $P_0$ and $P_1$ are their first and second halves of length $k$). The variables $P_0$ and $P_1$ correspond to even and odd blocks respectively. They may be dependent, and their distributions may differ from the initial distribution $P$ for $k$-bit blocks. Still we know that $P$ is the average of $P_0$ and $P_1$ since $P$ is computed for all blocks, and $P_0$/$P_1$ correspond to odd/even blocks. A convexity argument (the function $p\mapsto -p\log p$ used in the definition of entropy has negative second derivative) shows that $H(P)\ge [H(P_0)+H(P_1)]/2$.\footnote{There is more conceptual way to explain this: consider a random bit $b$ and random variable $P'$ that has the same distribution as $P_0$ when $b=0$ and the same distribution as $P_1$ when $b=1$. Then $P'$ has the same distribution as $P$, and $H(P')=H(P)$ is not smaller than $H(P'\cnd b)=[H(P_0)+H(P_1)]/2.$} Then
\[
H(P_0P_1)\le H(P_0)+H(P_1)\le 2H(P),
\]
 so $P_0P_1$ has twice bigger gap between entropy and length (at least). Repeating this argument, we can find $k$ such that the difference between length and entropy is greater than $1$. 

Now the second problem: what to do if some values of $P$ are zeros (some blocks have zero probability in the limit distribution). In this case we cannot use the code provided by Lemma~\ref{lem:shannon-fano}, since some blocks, while having zero limit frequency, still appear in the prefixes of $\alpha$ and have no codeword. Their limit frequency is zero, so it is not important how long would be the corresponding codewords, but some codewords for them are needed.

There are several ways to overcome this obstacle. For example, one may change the code provided by the Lemma, adding leading $0$ to all codewords, and use codewords starting from $1$ to encode ``bad'' blocks (having zero limit probabilities). Then all blocks, including the bad ones, will have codes of finite length, the constant $1$ in the Lemma is replaced by $2$ (this does not hurt), and we can proceed as before. 

 The other possibility is to use some $P'$ that is close to $P$ and has all non-zero probabilities. Then the limit average length of code will be bigger, since we use the code for $P'$ while the actual (limit) distribution is $P$. The overhead (called the Kullback -- Leibler distance between $P$ and $P'$) can be made arbitrarily small due to continuity, and we still get a contradiction making the overhead smaller than the gap.
 
This finishes the proof in one direction.

Now we need to prove that an arbitrary normal sequence $\alpha$ is incompressible. Let $R$ be an arbitrary automatic description mode. Consider some $k$ and split the sequence into $k$-bit blocks $\alpha=A_0A_1 A_2\ldots$. (Now $A_i$ are just the blocks in $\alpha$, not random variables). We will show that $$\liminf \KS_R(A_0 A_1\ldots A_{n-1})/nk$$ cannot be much smaller than~$1$. More precisely, we will show that
$$\liminf \frac{\KS_R(A_0 A_1\ldots A_{n-1})}{nk}\ge1-\frac{O(1)}{k},$$ 
where the constant in $O(1)$ does not depend on $k$. This is enough, because (i)~adding the last incomplete block can only increase the complexity and the change in length is negligible, and (ii)~the value of $k$ may be arbitrarily large. 
 
Now let us prove this bound for some fixed~$k$. Recall that $$\KS_R(A_0 A_1\ldots A_{n-1})\ge \KS_R(A_0)+\KS_R(A_1)+\ldots+\KS_R(A_{n-1})$$ and that $\KS(x)\le \KS_R(x)+O(1)$ for all $x$ and some $O(1)$-constant that depends only on $R$ (Theorem~\ref{th:complexity}). By assumption, all $k$-bit strings appear with the same limit frequency among $A_0$, $A_1$,\ldots, $A_{n-1}$. It remains to note that the average Kolmogorov complexity $\KS(x)$ of all $k$-bit strings is $k-O(1)$; indeed, the fraction of $k$-bit strings that can be compressed by more than $d$ bits ($\KS(x)<k-d$) is at most $2^{-d}$, and the series $\sum d 2^{-d}$ (the upper bound for the average number of bits saved by compression) has finite sum.

Alternatively, one may also note that for $d=\log k$ we have $O(1/k)$ fraction of strings that are compressible more than by $d$ (and at most by $k$) bits, and all other strings are compressible at most by $d=\log k$ bits, so the average compression is $O(1)+O(\log k)=O(\log k)$, and $O(\log k)$ bound is enough for our purposes (we do not need the stricter $O(1)$ bound proven earlier).
\end{proof}

A basic result of algorithmic information theory (Schnorr--Levin complexity characterization of randomness) says that for a Martin-L\"of random sequence $a_0a_1a_2\ldots$ we have $\KP(a_0\ldots a_{n-1})\ge n -O(1)$ (where $\KP(x)$ stand for the prefix complexity of $x$). Since prefix and plain complexity of $x$ differ by $O(\log |x|)$, we conclude that $\KS(a_0\ldots a_{n-1})\ge n-o(n)$ for all Martin-L\"of random sequences. Theorem~\ref{th:complexity}, (d) implies that the same is true for automatic complexities; therefore, according to our criterion (Theorem~\ref{th:compress}), every Martin-L\"of random sequence is normal (a classical result of algorithmic information theory, an effective version of the law of large numbers). Recalling that almost all bit sequences are Martin-L\"of random, we conclude that almost all bit sequences are normal (one of the first Borel's results about normality,~\cite[p.~261]{borel1909}). 

\section{Using the incompressibility criterion for normality}\label{sec:applications}

In this section we use the incompressibility characterization of normality to provide simple proofs for several classical results about normal sequences. First we prove that one may consider non-aligned frequencies of blocks when defining normality (Section~\ref{subsec:non-aligned}). Then we give a simple proof of Piatetski-Shapiro's theorem from~\cite{piatetski-shapiro,piatetski-shapiro1957} (Section~\ref{subsec:piatetski-shapiro}). We give a simple sufficient condition for the normality of Champernowne-type sequences (Theorem~\ref{th:normality-condition}, Section~\ref{subsec:champernowne}). This condition implies the normality of Champernowne's sequence; it is then applied to provide simple proofs of the results from Copeland--Erd\"os~\cite{copeland-erdos} (Section~\ref{subsec:copeland-erdos}), Besicovitch~\cite{besicovitch} (Section~\ref{subsec:besicovitch}) and Calude -- Staiger -- Stephan (Section~\ref{subsec:csr}).

\subsection{Non-aligned version of normality}\label{subsec:non-aligned}

Recall the proof of Theorem~\ref{th:compress}. A small modification of this proof adapts it to the non-aligned definition of normality, thus providing the proof of the equivalence between aligned and non-aligned definitions of normality. Let us see how this is done.

 Let $\alpha$ be a sequence that is not normal in the non-aligned version. This means that for some $k$ the (non-aligned) $k$-bit blocks do not have the correct limit distribution. These blocks can be split into $k$ groups according to their starting positions modulo~$k$. In one of the groups blocks do not have a correct limit distribution (otherwise the average distribution would be correct, too). So we can delete some prefix (less than $k$ symbols) of our sequence and get a sequence that is not normal in the aligned sense. Theorem~\ref{th:compress} says that its prefixes are compressible. The same is true for the original sequence since adding a fixed finite prefix (or suffix) changes complexity and length at most by~$O(1)$, after a suitable change of the description mode, as Theorem~\ref{th:complexity}, (a,b),  implies.

In the other direction the proof goes as follows (see the next paragraph for details). Let us assume that the sequence is normal in the non-aligned sense. The aligned frequency of some compressible-by-$d$-bits block (as well as any other block) can be only $k$ times bigger than its non-aligned frequency, which is exponentially small in $d$ (the number of saved bits), so we can choose the parameters to get the required bound.

Here are the details. Consider a non-aligned normal sequence, i.e., a sequence that does not satisfy the non-aligned version of normality definition. Now consider all blocks (strings) of length $k$ that are $d$-compressible in the sense that their $\KS_R$-complexity is smaller than $k-d$. There is at most $O(2^{k-d})$ of them, as Theorem~\ref{th:complexity}(e) says.  So their frequency among \emph{aligned} blocks in our sequence is at most $k2^{-d+O(1)}$.  Indeed, it can be only $k$ times bigger than their frequency among \emph{non-aligned} blocks: the numerator (the number of bad occurences) can only decrease, and the denominator (the total number of occurences) becomes $k$ times smaller if we consider only aligned occurences. 

For all non-$d$-compressible blocks $R$-compression saves at most $d$ bits, and for $d$-compressible blocks it saves at most $k$ bits, so the average number of saved bits (per $k$-bit block) is bounded by
$$
d+k2^{-d+O(1)}\cdot k = d +O(k^2 2^{-d}).
$$
We need this bound to be $o(k)$, i.e., we need that
$$
\frac{d}{k}+O(k2^{-d})=o(1)
$$
as $k\to\infty$. This can be achieved, for example, if $d=2\log k$. 

In this way we get the following corollary~\cite{borel1909,pillai1940,niven-zuckerman1951}:
\begin{corollary}
The aligned and non-aligned definitions of normality are equivalent.
\end{corollary}

Note also that adding/deleting a finite prefix does not change the compressibility, and, therefore, normality. (For the non-aligned version of the normality definition it is obvious anyway, but for the aligned version it is not so easy to see directly, see the discussion of the original Borel's definition of normality above for historical details.)

\subsection{Piatetski-Shapiro theorem}\label{subsec:piatetski-shapiro}

Piatetski-Shapiro in~\cite{piatetski-shapiro} proved\footnote{The proof used ergodic theory. Later in~\cite{postnikov1952,postnikov-piatetski1957,piatetski-shapiro1957} alternative proofs that do not refer to ergodic theory were given.} the following result: \emph{if for some constant $c$ and for all $k$ every $k$-bit block appears in a sequence $\alpha$ with \textup(non-aligned\textup) $\limsup$-frequency at most $c2^{-k}$, then the sequence $\alpha$ is normal}. 

This result is an immediate byproduct of the proof of the normality criterion (Theorem~\ref{th:compress}). Indeed, in the argument above we had a constant factor in the $O(k2^{-d})$ bound of Section~\ref{subsec:non-aligned} for the average compression due to compressible blocks. If the compressible blocks appear at most $c$ times more often (as well as all other blocks, but this does not matter), we still have the same $O$-bound, so we get Piatetski-Shapiro's result (in aligned and non-aligned version at the same time; Piatetski-Shapiro considered the aligned version).
 
We can even allow the constant $c$ to  depend on $k$ if its growth as a function of $k$ is not too fast. Namely, the following stronger result was proven by Piatetski-Shapiro in~\cite{piatetski-shapiro1957}): 

\begin{theorem}[Piatetski-Shapiro theorem, strong version]\label{th:piatetski}
Let $\alpha$ be an infinite bit sequence. Assume that for every $k$ and for every $k$-bit block $B$ its aligned \textup(or non-aligned\textup) frequency in all sufficiently long prefixes of $\alpha$ does not exceed $c_k 2^{-k}$, where $c_k$ depends only on $k$ and $c_k=2^{o(k)}$. Then $\alpha$ is a normal sequence.
\end{theorem}

\begin{proof}
Note first that the non-aligned version of this result follows from the aligned version. Indeed, aligned frequency of arbitrary block may exceed the non-aligned frequency of the same block only by a factor of $k$, and $k=2^{o(k)}$, so this additional factor still keeps $c_k=2^{o(k)}$.

To prove the aligned version of the result, recall the proof of Theorem~\ref{th:compress}. Consider some threshold $d_k$ (to be chosen later). We split all $k$-bit blocks into two groups: the blocks that are compressible by more than $d_k$ bits, and all the other ones. The fraction of the blocks of the first type (called ``compressible'' blocks in the sequel) among all $k$-bit strings is at most $2^{-d_k}$. Therefore, by the assumption, if we split a long prefix of $\alpha$ into aligned $k$-bit blocks, the fraction of compressible blocks among them is bounded by (approximately) $c_k 2^{-d_k}$, and each of them is compressible by at most $k$ bits (for obvious reasons). All other blocks are compressible by at most $d_k$ bits, so the number of saved bits per block is at most $kc_k2^{-d_k}+d_k$. We need this amount to be $o(k)$ to finish the proof of normality as before, so we need to choose $d_k$ in such a way that
\[
d_k =o(k) \quad \text{and}\quad kc_k2^{-d_k} = o(k).
\]
The second condition says that $d_k$ exceeds $\log c_k$ and the difference tends to infinity. Since $c_k=2^{o(k)}$, one can easily satisfy both conditions (e.g., let $d_k$ be $\log c_k+\log k$).
\end{proof}

\begin{remark}
In fact, Piatetski-Shapiro's statement in~\cite{piatetski-shapiro1957} is a bit stronger: it assumes only that $c_k=2^{o(k)}$ for infinitely many $k$, i.e., that $\liminf_k \frac{\log c_k}{k}=0$. The proof remains the same (we use the assumption to get a lower bound for automatic complexity of a prefix by splitting it into blocks; it is enough to do this not for all $k$ but for infinitely many $k$). Also there is a minor technical difference: Piatetski-Shapiro considered real numbers $\alpha$ and the distribution of fractional parts of $\alpha q^k$, where $q$ is the base (we consider the case $q=2$, but this does not matter). The condition in~\cite{piatetski-shapiro1957} says that the density of fractional parts that fall inside some interval $\Delta\subset [0,1]$ is bounded by $f(|\Delta|)$ for a suitable $f$, where $|\Delta|$ is the length of $\Delta$. It is easy to see that one can consider only intervals $\Delta$ obtained by dividing $[0,1]$ into $q,q^2,\ldots$ parts (since any interval can be covered by these ``aligned'' intervals with bounded overhead). In this way we get the statement formulated above.\footnote{We go into all these details since the paper~\cite{piatetski-shapiro1957} is published (in Russian) in a quite obscure place: a volume in a series published by Moscow Pedagogical Institute. It seems that this volume is now (June 2019) missing even in the library of the very institute that published it (now it is called Moscow Pedagogical State University). Fortunately, this volume is available in the Russian State Library in Moscow (though it is included only in the paper cards version of the catalog, not in the electronic database).}
\end{remark}

\begin{remark}
The bound for $c_k$ in this theorem is optimal, as the following example (from~\cite{piatetski-shapiro1957}) shows. Consider a sequence $\alpha$ that is random with respect to Bernoulli measure with parameter $\frac{1}{2}+\delta$. Then the frequency of the most frequent $k$-bit block (all ones) is $(\frac{1}{2}+\delta)^k=2^{-k}2^{\eps k}$ for some constant $\eps$ that can be arbitrarily small if $\delta$ is small. On the other hand, $\alpha$ is not normal.
\end{remark}

Let us note that Piatetski-Shapiro's result easily implies a result of Wall~\cite{wall}. Recall that a real number is normal if its binary expansion is normal. We ignore the integer part (since it has only finitely many digits, adding it as a prefix would not matter anyway).

\begin{theorem}[Wall~\cite{wall}]\label{th:wall}
If $p$ and $q$ are rational numbers and $\alpha$ is normal real number, then $\alpha p + q$ is normal.
\end{theorem}

\begin{proof} 
It is enough to show that a normal number remains normal when multiplied or divided by an integer (adding integers preserves normality for trivial reasons). Let $N$ be some integer factor. Fix some positions $m,m+1,\ldots,m+k-1$ in the binary expansion. Look at the digits of reals $\alpha$ and $N\alpha$ that occupy these positions. They form two $k$-bit blocks, one for $\alpha$ and one for $N\alpha$. Knowing the first one, we have $N$ possibilities for the second one (a school division algorithm keeps remainder modulo $N$), and vice versa (multiplication by $N$ also has this property). So if $\alpha$ is normal and frequencies of blocks in $\alpha$ are correct, in $N\alpha$ (or $\alpha/N$) each block appears at most $N$ times more often. It remains to apply Piatetski-Shapiro's theorem.
\end{proof}

\begin{remark}
Wall's theorem also can be derived from the characterization of normality in terms of automatic complexity (Theorem~\ref{th:compress}), since division and multiplication are automatic transformations.

See also below Theorem~\ref{th:doty} for a more general result of Doty, Lutz and Nandakumar saying that finite-state dimension does not change when a real number is multiplied by a rational factor.
\end{remark}

\subsection{A sufficient condition for normality in terms of complexity}\label{subsec:champernowne}

As we have mentioned, Champernowne~\cite{champernowne} proved that  the concatenation of the positional representations of all integers (in increasing order) is a normal sequence. (He considered decimal representations, not binary, but this does not make any difference.) 

This result is a special case of the following simple observation, a sufficient condition for normality in terms of Kolmogorov complexity. 

\begin{theorem}\label{th:normality-condition}
Let $x_1,x_2,x_3,\ldots$ be a sequence of non-empty binary strings. Let $L_n$ be a rational number that is the average length of $x_1,\ldots,x_n$, i.e., $L_n= (|x_1|+\ldots+|x_n|)/n$. Let $C_n$ be their average Kolmogorov complexity, i.e., $C_n=(\KS(x_1)+\ldots+\KS(x_n))/n$. Assume that $|x_n|=o(|x_1|+\ldots+|x_{n-1}|)$ and $L_n\to\infty$ as $n\to\infty$.

If $C_n/L_n \to 1$ as $n\to\infty$, then the concatenated sequence $\varkappa=x_1x_2x_3\ldots$ is normal.
\end{theorem}

The first two assumptions are technical (and usually are easy to check); they guarantee that $|x_n|$ grows not too slow and not too fast. In this case the normality of concatenation is guaranteed if the average \emph{complexity} of strings $x_i$ is close to their average \emph{length}. Note that $C_n$ is defined up to $O(1)$ additive term (the complexity function is defined with the same precision) and that $C_n\le L_n+O(1)$.
 
\begin{proof}
Using Theorem~\ref{th:compress}, we need to prove, for an arbitrary fixed automatic description mode $R$, a  lower bound $N-o(N)$ for the automatic complexity of the $N$-bit prefix of $x_1x_2x_3\ldots $. This prefix may end inside some $x_i$; we ignore the last incomplete block and consider maximal prefix of the form $x_1\ldots x_M$ of length at most $N$. Due to the superadditivity property (Theorem~\ref{th:complexity}, (k)) the automatic complexity of the $N$-bit prefix is at least $\KS_R(x_1)+\ldots+\KS_R(x_M)$ and is at least $\KS(x_1)+\ldots+\KS(x_M)-O(M)$, since Kolmogorov complexity is a lower bound for automatic complexity up to $O(1)$ additive term.

Due to the assumption $|x_n|=o(|x_1|+\ldots+|x_{n-1}|)$, the ignored incomplete part has length $o(N)$, so we may replace $N$ in the desired lower bound by $|x_1|+\ldots+|x_M|$. It remains to note that the ratio
\[
\frac{\KS_R(x_1)+\ldots+\KS_R(x_M)}{|x_1|+\ldots+|x_M|}\ge
\frac{\KS(x_1)+\ldots+\KS(x_M)-O(M)}{|x_1|+\ldots+|x_M|}=
\frac{C_M-O(1)}{L_M}
\]
converges to $1$ according to our assumptions ($L_M\to\infty$ and $C_M/L_M\to 1$). Here in the last step we divided the numerator and the denominator by $M$.
\end{proof}

We formulated Theorem~\ref{th:normality-condition} for the binary case, but both the statement and the proof can be easily adapted to an arbitrary base.

For the Champernowne example $x_i$ is the binary representation of $i$. The average length of $x_1,\ldots,x_n$, and even the maximal length,  is obviously bounded by $\log n + O(1)$. As for the complexity, it is enough to note that all $x_i$ are different, and the number of different strings of complexity at most $u$ is $O(2^u)$. Therefore, the fraction of strings that have complexity at most $\log n-d$ among all strings $x_1,\ldots,x_n$ is $O(2^{-d})$. The series $\sum d2^{-d}$ converges, so the average complexity of $x_1,\ldots,x_n$ is at least $\log n - O(1)$, and $C_n/L_n\ge (\log n -O(1))/(\log n+O(1))\to 1$. Other conditions of Theorem~\ref{th:normality-condition} are obviously true.

Champernowne's paper~\cite{champernowne} contains some other results: Theorem I says that the concatenation of all \emph{strings} in order of increasing lengths, i.e., the sequence 
\[
 0 \, 1 \, 00\, 01\, 10\, 11\, 000\,001\,010\,\ldots
\]
is normal. Theorem~II says that it remains normal if every string is repeated $\mu$ times for some integer constant~$\mu$. Theorem~III is the Champernowne's example we started with. Theorem~IV considers a sequence where $i$th string is repeated $i$ times. In all these examples the normality obviously follows from our Theorem~\ref{th:normality-condition}. (For Theorem~IV we need to note that a subset that has measure at most $2^{-d}$ according to the uniform distribution on $\{1,2,\ldots,n\}$ has measure $O(2^{-d})$ if we change the distribution and let the probability of $i$ be proportional to $i$.) 

\subsection{Copeland -- Erd\"os theorem}\label{subsec:copeland-erdos}

In addition to Theorems~I--IV (see the previous section) Champernowne~\cite{champernowne} gave some other examples of normal numbers (sequences), saying that they ``need for their establishment tedious lemmas and an involved notation, [and] no attempt at a proof will be advanced''. These examples are the sequences made of concatenated representations of  (a)~all composite numbers, (b)~numbers $\lfloor \alpha n \rfloor$ for some positive real $\alpha$, and (c)~$\lfloor n \log n \rfloor$. In all these cases the normality easily follows from Theorem~\ref{th:normality-condition}.

Champernowne also stated as a conjecture that the sequence made of decimal representation of \emph{prime} numbers is normal. This conjecture was proven by Copeland and Erd\"os~\cite{copeland-erdos} who gave a sufficient condition for the sequence $x_1x_2x_3\ldots$ obtained by concatenating the positional representations of integers $x_1,x_2,\ldots$ to be normal. Let us state the Copeland -- Erd\"os theorem and show that it is a direct consequence of Theorem~\ref{th:normality-condition}.

\begin{theorem}[Copeland -- Erd\"os]
Let $x_1,x_2,x_3,\ldots$ be a strictly increasing sequence of integers, and the number of terms $x_i$ that are less than $2^m$ is at least $2^{m(1-o(1))}$. The sequence $x_1 x_2 x_3\ldots$ \textup(the concatenation of the positional representations\textup) is normal.
\end{theorem}

\begin{proof}
The assumption implies that the length of $x_n$ is $(1+o(1))\log n$, and the conditions for the lengths are true for obvious reasons. The lower bound for complexity in the Champernowne example (Section~\ref{subsec:champernowne}) used only that all $x_i$ are different, and the assumption guarantees that $x_i$ form a strictly increasing sequence and therefore are different. So we may apply Theorem~\ref{th:normality-condition}.
\end{proof} 

\subsection{Besicovitch's theorem}\label{subsec:besicovitch}

Besicovitch~\cite{besicovitch} has proven that the number obtained by the concatenation of all perfect squares in the increasing order is normal. This result is also a consequence of Theorem~\ref{th:normality-condition}, but more detailed analysis is needed. We give a sketch of the corresponding argument. Applying Theorem~\ref{th:normality-condition}, we let $x_i$ be the binary representation of $i^2$ (we will say later what changes are needed for decimal representation). The length of $x_i$ is about $2\log i$, while the (typical) complexity is the complexity of $i$, so $C_n/L_n$ is close to $1/2$, not $1$. To deal with this problem, let us divide the string $x_i$ in two halves of equal length $x_i=y_iz_i$ and consider the most significant half and the least significant half separately. Of course, if $x_i$ has odd length, then the lengths of $y_i$ and $z_i$ differ by~$1$, and the reasoning should be adapted. We do not go into these details.

Instead of using $\KS(x_i)$ as a lower bound for $\KS_R(x_i)$, we note that $\KS_R(x_i)\ge \KS_R(y_i) +  \KS_R(z_i)$ and then use $\KS(y_i)$ and $\KS(z_i)$ as the lower bounds for both summands. In other words, we apply our criterion to a sequence of left and right halves of the binary string $x_i$. For $y_i$ we note that the most significant half of the binary representation of $i^2$ determines $i$ almost uniquely (there are $O(1)$ possible values of $i$ with the same most significant half). Indeed, assume that $i$ is a $k$-bit number. How much can we change $i$ not changing the most significant half of $i^2$? Changing $i$ by~$1$, we change $i^2$ by $2i+1$, and this change is of order $2^k$ since $i$ is a $k$-bit numbers. Only $O(1)$ changes of this type can be made without changing the most significant half of $i^2$ (i.e., the $k$ most significant bits out of $2k$). There is a caveat here: one should also take into consideration the possibility that two halves have different lengths (by~$1$), but this gives only $O(1)$ new candidates.

For the least significant half $z_i$ more complicated analysis is needed, since $z_i$ does not determine $i$ and sometimes many different values of $i$ share the same $z_i$. For example, if $i$ is a $k$-bit number whose $k/2$ least significant bits are zeros, then $i^2$ is a $2k$-bit number with $k$ trailing zeros, so we have about $2^{k/2}$ different values of $i$ that share the same $z_i$ (all zeros). This happens rarely, as the following lemma shows:

\begin{lemma}
 For each $k$, the average Kolmogorov complexity of $x^2 \bmod 2^k$ taken over all $x$ modulo $2^k$ is $k-O(1)$.
\end{lemma} 

\begin{proof}[Proof sketch]
As we mentioned, complexity $\KS(x^2)$ can be much less than $\KS(x)$ (if $x$ ends with $k/2$ zeros, the complexity of $x^2$ is zero while the complexity of $x$ could be $k/2$). However, such a difference is possible only if $x$ ends with many zeros. More precisely, we have $\KS(x^2\bmod 2^k)\ge \KS(x)-O(\zeta(x))$ for $k$-bit string $x$, where $\zeta(x)$ is the number of trailing zeros in $x$ (the maximal $u\le k$ such that $2^u$ divides $x$). This is enough, since the expected value of $\zeta(x)$ for random $x$ modulo $2^k$ is $O(1)$ (half of all numbers have at least one trailing zero, half of those have at least one additional trailing zero, etc., and the series converges).

To prove the bound let us rewrite it as $\KS(x)\le \KS(x^2)+O(\zeta(x))$. To specify $x$, it is enough to specify $x^2$ and the ordinal number of $x$ in the set of all residues with the same square. Therefore, it is enough to show that the number of residues $y$ modulo $2^k$ such that $x^2=y^2 \pmod {2^k}$ is bounded by $2^{O(z)}$ if $x$ has $z$ trailing zeros. Indeed, assume that $x$ has $z$ trailing zeros and $x^2=y^2$ for some other $y$ modulo $2^k$. Then $x^2-y^2=(x-y)(x+y)$ is a multiple of $2^k$, therefore $x-y$ is a multiple of $2^u$ and $x+y$ is a multiple of $2^v$ for some $u,v$ such that $u+v=k$. Then $2x$ is a multiple of $2^{\min(u,v)}$, so $\min(u,v)\le z+1$. Then $\max(u,v)\ge k - z + O(1)$ (recall that $\min(u,v)+\max(u,v)=u+v$), so one of $x-y$ and $x+y$ is a multiple of $2^{k-z+O(1)}$, and each case contributes at most $2^{z+O(1)}=O(2^z)$ solutions for the equation $x^2=y^2 \pmod {2^k}$.
\end{proof}

\begin{remark}
The statement of the lemma involves Kolmogorov complexity, but it would be enough to get a lower bound for the Shannon entropy of $x^2 \bmod 2^k$ where $x$ is a uniformly distributed random integer modulo $2^k$. Indeed, assume that we have such a lower bound. Then we can derive the lower bound for the average Kolmogorov complexity of the squares of $2^k$ first integers, if we use the prefix version of the complexity (see, e.g.,~\cite{shen-intro} for the definition and properties of prefix complexity). Indeed, in this case the optimal prefix-free description of these $2^k$ strings form a prefix-free code, and the average length of a prefix-free code for a random variable is at least the entropy of this variable. One can also note that a upper bound for Shannon entropy can be converted for the lower bound for the Kolmogorov complexity, since the sequence is computable and the optimal prefix-free code for it is also computable, so the connection works in both directions, and we may use the entropy language if we want to.
\end{remark}

See also Section~\ref{subsec:stateless} for the dimension version of Theorem~\ref{th:normality-condition} (Theorem~\ref{th:lower-bound-dimension}).

\begin{question}
Is it possible to generalize these arguments and prove the  Davenport -- Erd\"os result (replace the squaring in Besicovitch's theorem by a polynomial of higher degree)? One possible approach would be to estimate the entropies of the random variables obtained as follows: fix some $t$, take a random integer $x\in \{0,\ldots,2^k-1\}$, compute $P(x)$, and let $\xi_t$ be the bit string that consists of $k$ consecutive bits in the binary representation of $P(x)$, starting from position $t$.
\end{question}

\subsection{Calude -- Salomaa -- Roblot question answered by Calude -- Staiger -- Stephan}\label{subsec:csr}

In this section we use our tools to give a simple answer to a question posed by Calude, Salomaa and Roblot~\cite[Section 6]{csr} and answered in~\cite{css} by a more complicated argument. In~\cite{csr} the authors define a version of automatic complexity in the following way. A deterministic transducer (finite automaton that reads an input string and at each step produces some number of output bits) maps a description string to a string to be described, and the complexity of $y$ is measured as the minimal sum of the sizes of the transducer and the input string needed to produce $y$; the minimum is taken over all pairs (transducer, input string) producing $y$. The size of the transducer is measured via some encoding, so the complexity function depends on the choice of this encoding. ``It will be interesting to check whether finite-state random strings are Borel normal''~\cite[p.~5677]{csr}. Since normality is defined for infinite sequences, one probably should interpret this question in the following way: is it true that normal infinite sequences can be characterized as sequences whose prefixes have finite-state complexity close to length?

This question got a negative answer in~\cite{css}. Here we show that the our tools can be used to provide a simple proof of this negative answer. More precisely, in one direction this approach works, but in the other direction it fails. To avoid confusion between different versions of automatic complexity, we denote the complexity defined in~\cite{csr} by $\CSR(x)$. It depends on the choice of the encoding of transducers, but the claim is true for every encoding, so we assume that some encoding is fixed and omit it in the notation.

\begin{theorem}[\cite{css}]\label{th:csr-answer}
\leavevmode\par
\textup{\textbf{(a)}} If a binary sequence $\alpha=a_0a_1\ldots$ is not normal, then there exist some $c<1$ such that the $\CSR(a_0\ldots a_{n-1}) < cn$ for infinitely many $n$.

\textup{\textbf{(b)}} There exists a normal binary sequence $\beta=b_0b_1\ldots$ such that \[\liminf \CSR(b_0\ldots b_{n-1})/n=0.\]

\end{theorem}

\begin{proof}
To prove the first statement, we repeat the argument used to prove the first part of Theorem~\ref{th:compress}. Indeed, the block code constructed in that argument can be decoded by a transducer. This transducer had some description of fixed length, and then we add the length of the encoded string. For long prefixes the transducer part does not matter, since the transducer is fixed and the length of the prefix goes to infinity.

For the second part we construct an example of a normal sequence using the Champernowne's idea and Theorem~\ref{th:normality-condition}. The sequence will have the form
\[
\beta=(B_1)^{n_1}(B_2)^{n_2}\ldots
\]
Here $B_i$ is the concatenation of all strings of length $i$ (say, in lexicographical ordering, but this does not matter), and $n_i$ is a fast growing sequence of integers.

To choose $n_i$, let us note first that for a periodic sequence (of the form $XY^\infty$) the $\CSR$-complexity of its prefixes of the form $XY^k$ is $o(\text{length of }XY^k)$. Indeed, we may consider a transducer that first outputs $X$, then outputs $Y$ for each input bit $1$. So $\CSR(XY^m)=m+O(1)$, and the compression ratio is about $1/|Y|$. To get an $o(\text{length})$-bound, we use $Y^c$ for some constant $c$ as a period to improve the compression.

Now consider the complexity/length ratio for the prefixes of $\beta$ if the sequence $n_i$ grows fast enough. Assume that $n_1,n_2,\ldots,n_k$ are already chosen and we now choose the value of $n_{k+1}$. We may use the bound explained in the previous paragraph and let $X=(B_1)^{n_1}\ldots (B_k)^{n_k}$ and $Y=B_{k+1}$. For large enough $n_{k+1}$ we get arbitrarily small complexity/length ratio. (Note that good compression is guaranteed only for some prefixes; when increasing $k$, we need to switch to another transducer, and we know nothing about the length of its encoding. This corresponds to $\liminf$ in our statement.)

It remains to apply Theorem~\ref{th:normality-condition} to show that for fast growing sequence $n_1,n_2,\ldots$ the sequence $\beta$ is normal. We apply the criterion by splitting $B_k$ into pieces of length $k$ (so all strings of length $k$ appear once in this decomposition of $B_k$). We already know that the average Kolmogorov complexity of the pieces in $B_k$ is $k-O(1)$ (and the length of all pieces is $k$). This is enough to satisfy the conditions from Theorem~\ref{th:normality-condition} \emph{if $x_1\ldots x_n$ ends on the boundary of the block $B_k$}. But in general we need also to consider the last incomplete group of blocks that form a prefix of some $B_k$. The total length of these blocks is bounded by $|B_k|$, i.e., by $k2^k$. We need this group to be short compared to the rest, and this will be guaranteed if $n_{k-1}$ (the lower bound for the length of the previous part) is much bigger than $k2^k$. And we assume that $n_k$ grow very fast, so this condition is easy to satisfy. Theorem~\ref{th:csr-answer} is proven.
\end{proof}

\section{Finite-state dimension and automatic complexity}\label{sec:fsd}

If a sequence is not normal, we may ask how far it is from being normal. This is measured by the notion of finite-state dimension introduced in~\cite{dai-fsd}. This is a ``finite-state version'' of the notion of effective dimension (as we have discussed in the Introduction). The finite-state dimension of a binary sequence is a number between $0$ and $1$; it is an upper bound for the effective Hausdorff dimension. Finite-state dimension equals $1$ for normal sequences. In this section we extend some results of Sections~\ref{sec:normal-incompressible} and~\ref{sec:applications} proven for normal sequences to the case of arbitrary finite-state dimension, and discuss the connections between the finite-state dimension and the effective Hausdorff dimension.

 We start (Section~\ref{subsec:entropy-rates}) by defining the finite-state dimension in terms of entropy rates for aligned blocks, following Bourke, Hitchcock and Vinodchandran~\cite{bourke} who proved that this definition is equivalent to the original one  given in terms of finite-state gales~\cite{dai-fsd}. We prove that one may as well use the non-aligned blocks in this definition. In the next section (Section~\ref{subsec:large-blocks}) we show that the equivalence between aligned and non-aligned blocks in the definition of finite-state dimension requires a change in the block size. Then (Section~\ref{subsec:fsd-wall}) we give a simplified proof of the result of Doty, Lutz and Nandakumar~\cite{doty} saying that the finite-state dimension of a real number remains the same when the number is multiplied by a rational number (a dimension version of Wall's theorem), improving the bound for entropies of $k$-bit blocks, and give a simple example showing that this bound is tight. Then we prove the characterization of finite-state dimension in terms of automatic complexity (Section~\ref{subsec:fsd-ac}, Theorem~\ref{th:fsd-automatic}). Moreover, as a byproduct we get (Section~\ref{subsec:stateless}) a ``stateless'' characterization of finite-state dimension that does not mention at all finite-state automata or Shannon entropy and uses superadditive upper bounds for Kolmogorov complexity (Theorem~\ref{th:stateless}).  We give also another stateless characterization that uses some ``calibration'' condition instead of Kolmogorov complexity (Theorem~\ref{th:stateless1}).  Then we recall the original definition of the finite-state dimension in terms of finite-state $s$-gales (Section~\ref{subsec:fsm}) and use the tools from algorithmic information theory (a finite-state version of a priori probability) to give a simple proof of equivalence between this definition and the others. In Section~\ref{subsec:agafonov} we use martingales to provide simple proofs for the results of Agafonov (Theorem~\ref{th:agafonov}) and Schnorr -- Stimm (Theorem~\ref{th:s-and-s}). Finally, in Section~\ref{subsec:multi-account} we note that some more general notion of a finite-state measure can also be used to characterize finite-state dimension and normality.

\subsection{Entropy rates for aligned and non-aligned blocks}\label{subsec:entropy-rates}

Consider a sequence $\alpha=a_0a_1a_2\ldots$, and some positive integer $k$. As in the definition of normality, we cut the sequence $\alpha$ into $k$-bit blocks (aligned version), or consider all $k$-bit substrings of $\alpha$ (non-aligned version). Then we consider limit frequencies of these blocks. In this way we get some distribution on the set $\mathbb{B}^k$ of all $k$-bit blocks. We want to define the finite-state dimension of $\alpha$ as the limit of the Shannon entropy of this distribution per bit, i.e., divided by $k$, as $k$ goes to infinity.

The problem is that limit frequencies may not exist, so we should be more careful. For every $N$  take the first $N$ blocks of length $k$ and choose one of them uniformly at random.  In this way we obtain a random variable taking values in $\mathbb{B}^k$. Consider the Shannon entropy of this random variable.  This can be done in aligned (a) and non-aligned (na) settings, so we get two quantities:
$$H^\mathrm{a}_{k,N}(\alpha) = H(\alpha_{kI} \ldots \alpha_{kI + k - 1}), \qquad H^\mathrm{na}_{k,N}(\alpha) = H(\alpha_{I}\ldots \alpha_{I + k - 1}),$$
where $I\in\{0, \ldots, N - 1\}$ (the block number) is chosen uniformly at random, and $H$ denotes the Shannon entropy of the corresponding random variable.

Then we apply $\liminf_N$ as $N\to\infty$ and define
\[
H^\mathrm{a}_k(\alpha)=\liminf_{N\to \infty} H^\mathrm{a}_{k,N}(\alpha), 
 \qquad
H^\mathrm{na}_k(\alpha)=\liminf_{N\to \infty} H^\mathrm{na}_{k,N}(\alpha).
\]
The following result says that both quantities $H^\mathrm{a}_k(\alpha)$ and $H^\mathrm{na}_k(\alpha)$, divided by the block length  $k$, converge to the same value as $k\to\infty$, and this value can also be defined as $\inf_{k}H_k(\alpha)/k$ (both in aligned and non-aligned versions).

\begin{theorem}\label{th:fsd}
For every bit sequence $\alpha$ we have
\[
\lim_{k}\frac{H^\mathrm{a}_k(\alpha)}{k}=
\inf_{k}\frac{H^\mathrm{a}_k(\alpha)}{k}=
\lim_{k}\frac{H^\mathrm{na}_k(\alpha)}{k}=
\inf_{k}\frac{H^\mathrm{na}_k(\alpha)}{k}.
\]
The sequence $\alpha$ is normal if and only if this common value equals $1$.
\end{theorem}

\begin{definition}
This common value of these four quantities is called the \emph{finite-state dimension of~$\alpha$} and is denoted by $\dimFS(\alpha)$.
\end{definition}

The original definition of the finite-state dimension~\cite{dai-fsd} was different (see Section~\ref{subsec:fsm} below), and the equivalence between it and the aligned version of the definition given above was shown in~\cite{bourke}. See also Theorem~\ref{th:fsg} (part 2) below. The  equivalence between non-aligned and aligned versions seems to be new.

\begin{proof}
There are two ways to prove Theorem~\ref{th:fsd}. Here we give a proof that uses basic tools from information theory such as Shearer-type inequalities. One can also prove this result using the characterization of finite-state dimension in terms of automatic complexity. We sketch this proof later, see the remark at the end of Section~\ref{subsec:fsd-ac} (p.~\pageref{rem:alternative}).

The technical part of the proof consists of two lemmas:

\begin{lemma}
\label{lem:fsd1}
For every $\alpha$, every $k$, every $K\ge k$ \textup:
\[
\frac{H^\mathrm{na}_K(\alpha)}{K} \le  \frac{H^\mathrm{a}_k(\alpha)}{k} + O\left(\frac{k}{K}\right).
\]
\end{lemma}
\begin{lemma}
\label{lem:fsd2}
For every $\alpha$, every $k$, every $K\ge k$ \textup:
\[
\frac{H^\mathrm{a}_K(\alpha)}{K} \le  \frac{H^\mathrm{na}_k(\alpha)}{k} + O\left(\frac{k}{K}\right).
\]
\end{lemma}

Let us show how these two lemmas imply Theorem \ref{th:fsd}. Take the $\limsup$ of the both sides of these inequalities as $K\to\infty$:
\[\limsup_{K\to\infty}  \frac{H^\mathrm{na}_K(\alpha)}{K} \le \frac{H^\mathrm{a}_k(\alpha)}{k}, \qquad \limsup_{K\to\infty}  \frac{H^\mathrm{a}_K(\alpha)}{K} \le \frac{H^\mathrm{na}_k(\alpha)}{k}.\]
Since this holds for all $k$,
\[
\limsup\limits_{K\to\infty}  \frac{H^\mathrm{na}_K(\alpha)}{K} \le \inf\limits_k\frac{H^\mathrm{a}_k(\alpha)}{k}, \qquad \limsup\limits_{K\to\infty}  \frac{H^\mathrm{a}_K(\alpha)}{K} \le \inf\limits_k\frac{H^\mathrm{na}_k(\alpha)}{k}.
\]
Obviously we also have:
\[
\inf\limits_k\frac{H^\mathrm{a}_k(\alpha)}{k} \le \limsup\limits_{K\to\infty}  \frac{H^\mathrm{a}_K(\alpha)}{K},\qquad\inf\limits_k\frac{H^\mathrm{na}_k(\alpha)}{k} \le \limsup\limits_{K\to\infty}  \frac{H^\mathrm{na}_K(\alpha)}{K},
\]
so all four quantities coincide and are equal to both $\lim_{k}\dfrac{H^\mathrm{na}_k(\alpha)}{k}$ and $\lim_{k}\dfrac{H^\mathrm{a}_k(\alpha)}{k}$.

It remains to prove Lemmas \ref{lem:fsd1} and \ref{lem:fsd2}.

\begin{proof} [Proof of Lemma \ref{lem:fsd1}]
Fix a sequence $\alpha=a_0a_1a_2\ldots$, and consider some integer $N$. Take $I \in\{0, \ldots, N - 1\}$ uniformly at random and consider a random variable
$$\xi = a_I \ldots a_{I + K - 1}$$
whose values are $K$-bit strings. In other words, this random variable is a randomly selected non-aligned block among the first $N$ ones. By definition, the entropy of $\xi$ is $H^\mathrm{na}_{K,N}(\alpha)$. Let us look at the aligned $k$-bit blocks covered by the block $\xi$ (i.e., the aligned $k$-bit blocks inside $I\ldots I+K-1$). The exact number of these blocks may vary depending on $I$, but there are at least $m=\lfloor K/k\rfloor -1$ of them (if there were only $m-1$ complete blocks, plus maybe two incomplete blocks, then the total length would be at most $k(m-1)+2k-2=km+k-2$, but we have $K/k\ge m+1$, i.e., $K\ge km+k$).  We number $m$ first covered aligned blocks from left to right and get $m$ random variables $\xi_1,\ldots,\xi_m$ (defined at the same space $\{0,\ldots,N-1\}$). For example, $\xi_1$ is the leftmost aligned $k$-bit block of $\alpha$ in the interval $I\ldots I+K-1$. To reconstruct the value of $\xi$ when all $\xi_i$ are known, we need to specify the prefix and suffix of $\xi$ that are not covered by $\xi_i$ (including their lengths). This requires $O(k)$ bits of information, so
\[
H^\mathrm{na}_{K,N}(\alpha)  = H(\xi)\le H(\xi_1)+\ldots+H(\xi_m)+O(k).
\]
We will show that for each $s\in\{1,\ldots,m\}$ the distribution of the random variable $\xi_s$ is close to the uniform distribution over the first $\lfloor N/k\rfloor$ aligned $k$-bit blocks of $\alpha$. The standard way to measure how close are two distributions on the same set $X$ is to measure the \emph{statistical distance} between them, defined as 
\[
\delta(P, Q) = \frac{1}{2}\sum\limits_{x\in X} \bigl|P(x) - Q(x)\bigr|.
\]
We claim that (for each $s\in \{1,2,\ldots,m\}$) the statistical distance between the distribution of $\xi_s$ and the uniform distribution on the first $\lfloor N/k\rfloor$ aligned blocks converges to $0$ as $N\to \infty$. First, let us note that for a fixed aligned block its probability to become $s$-th aligned block inside a random nonaligned block is exactly $k/N$ (there are $k$ possible positions for a random non-aligned block when this happens). The only exception to this rule are aligned blocks that are near the endpoints, and we have at most $O(K/k)$ of them. When we choose a random aligned block, the probability to choose some position is exactly $1/\lfloor N/k\rfloor$, so we get some difference due to rounding. It is easy to see that the impact of both factors on the statistical distance converges to $0$ as $N\to\infty$. Indeed, the number of the boundary blocks is $O(K/k)$, and the bound does not depend on $N$, while the probability of each block (in both distributions) converges to zero.\footnote{More precisely, we should speak not about the probability of a given block, since the same $k$-bit block may appear in several positions, but about the probability of its appearance in a given position. Formally speaking, we use the following obvious fact: if we apply some function to two random variables, the statistical difference between them may only decrease. Here the function forgets the position of a block.} Also, since $m=N/k$ and $m'=\lfloor N/k\rfloor$ differ at most by $1$, the difference between $1/m$ and $1/m'$ is of order $1/m^2$, and converges to $0$ even if multiplied by $m$ (the number of blocks is about $m$).

Now we use the continuity (more precisely, the uniform continuity) of the entropy function and note that all $m=\lfloor N/k\rfloor -1$ random variables in the right hand side are close to the uniform distribution on first $\lfloor N/k\rfloor$ aligned blocks (the statistical distance converges to $0$), so  
\[
\liminf_{N\to\infty}H^\mathrm{na}_{K,N}(\alpha) \le (\lfloor K/k\rfloor -1) \liminf_{N\to\infty} H^{\mathrm{a}}_{k, \lfloor N/k\rfloor}(\alpha) +O(k),
\]
and dividing by $K$ we get the statement of Lemma~\ref{lem:fsd1}.
\end{proof}

\begin{proof}[Proof of Lemma \ref{lem:fsd2}]
We need an upper bound for $H^\mathrm{a}_{K,N}(\alpha)$, i.e.,  for $H(a_{KI} \ldots a_{KI + K - 1})$ where $I$ is uniformly distributed in $\{0, 1\ldots, N - 1\}$. For that we use Shearer's inequality (see, e.g.,~\cite[Section 7.2 and Chapter 10]{kolmbook}). In general, this inequality can be formulated as follows.  Consider a finite family of arbitrary random variables $\eta_0,\ldots,\eta_{m-1}$ indexed by integers in $\{0,\ldots,m-1\}$. For every $U\subset\{0,\ldots,m-1\}$ consider the tuple $\eta_U$ of all $\eta_u$ where $u\in U$. If a family of subsets $U_0,\ldots,U_{s-1}\subset \{0,\ldots,m-1\}$ covers each element of $U$ at least $r$ times, then
\[
H(\eta_U)\le \tfrac{1}{r}\left( H(\eta_{U_0})+\ldots+H(\eta_{U_{s-1}})\right).
\]
In our case we have $K$ variables $\eta_0,\ldots \eta_{K-1}$ that are individual bits in a random aligned $K$-bit block $a_{KI} \ldots a_{KI + K - 1}$ (for random $I$), i.e. $\eta_0=a_{KI}$, $\eta_1=a_{KI+1}$, etc. The set $U$ contains all indices $0,\ldots, K-1$, and the sets $U_i$ contains $k$ indices $i,i+1,\ldots ,i+k-1$ (where operations are performed modulo $K$, so there are $U_i$ that combine the prefix and suffix of a random $K$-bit block). Each $\eta_i$ is covered $k$ times due to this cyclic arrangement. In other words, the variable $\eta_{U_i}$ is a substring of the random string  $\eta_U=a_{KI} \ldots a_{KI + K - 1}$ that starts from $i$th position and wraps around if there is not enough bits. There are $k - 1$ tuples of this ``wrap-around'' type (block of length $k$ may cross the boundary in $k-1$ ways). These tuples are not convenient for our analysis, so we just bound their entropy by $k$. In this way  we obtain the following upper bound:
\[
H^\mathrm{a}_{K,N}(\alpha) =
H(a_{KI} \ldots a_{KI + K - 1}) \le\frac{1}{k}\left(\sum_{s = 0}^{K - k} H(a_{KI + s}\ldots a_{KI + s + k - 1}) + (k - 1) k \right).
\]
Adding $k-1$ terms (adding some other entropies that replace the ``wrap-around terms''), we increase the right hand side:
\[
H^\mathrm{a}_{K,N}(\alpha) \le \frac{1}{k}\left(\sum_{s = 0}^{K - 1} H(a_{KI + s}\ldots a_{KI + s + k - 1}) + (k - 1) k \right).
\]
Let us look at the variable $a_{KI + s}\ldots a_{KI + s + k - 1}$ in the right hand side for some fixed~$s$. It has the same distribution as the random non-aligned $k$-bit block $a_J\ldots a_{J+k-1}$ for uniformly chosen $J$ in $\{0,\ldots,NK-1\}$ conditional on the event ``$J \bmod K=s$'':
\[
H(a_{KI + s}\ldots a_{KI + s + k - 1}) = H(a_{J}\ldots a_{J + k - 1}\cnd J\bmod K = s).
\]
The average of these $K$ entropies (for $s=0,\ldots,K-1$) is the conditional entropy
\[
H(a_{J}\ldots a_{J + k - 1}\cnd J\bmod K)
\]
 that does not exceed the unconditional entropy. So we get
\[
H^\mathrm{a}_{K,N}(\alpha) \le \frac{1}{k} \left(K\cdot H^\mathrm{na}_{k,KN}(\alpha) + (k - 1) k \right).
\]
By taking $\liminf$ as $N\to\infty$ we obtain:
$$\frac{H^\mathrm{a}_{K}(\alpha)}{K} = \liminf\limits_{N\to\infty} \frac{H^\mathrm{a}_{K,N}(\alpha)}{K} \le \liminf\limits_{N\to\infty} \frac{H^\mathrm{na}_{k,KN}(\alpha)}{k} + O\left(\frac{k}{K}\right) .$$
However, $\liminf$ in the right hand side is taken over multiples of $K$ and we want it to be over all indices. Formally, it remains to show that
$$ \liminf\limits_{N\to\infty} \frac{H^\mathrm{na}_{k,KN}(\alpha)}{k} =  \liminf\limits_{N\to\infty} \frac{H^\mathrm{na}_{k,N}(\alpha)}{k}$$
as the latter is by definition equal to $H^\mathrm{na}_{k}(\alpha)/k$. Indeed, the statistical distance between the uniform distribution on the first $KN$ (non-aligned) blocks and the uniform distribution on the first $KN+r$ blocks (where $r$ the remainder modulo $K$) tends to zero since the first distribution is the second one conditioned on the event whose probability converges to $1$ (i.e., the event ``the randomly chosen block is not among the $r$ last ones'' whose probability is $KN/(KN+r)$).
\end{proof}

Theorem~\ref{th:fsd} is proven.
\end{proof}

\subsection{Why do we need large blocks: a counterexample}
\label{subsec:large-blocks}

In the previous section we have shown that the aligned finite-state dimension is equal to the non-aligned  one. However, this argument uses different block sizes: we show that if $H^\mathrm{a}_k(\alpha)/k$ is small, then $H^\mathrm{na}_K(\alpha)/K$ is small for much larger $K$, and vice versa. This change in the block size is unavoidable, as the following example shows (when $k=2$, no fixed $K$ is enough):

\begin{theorem}\label{thm:large-blocks}
\leavevmode
\begin{description}
\item[\textup{(a)}]For all  $k$ there exists an infinite sequence $\alpha$ such that  $H^\mathrm{na}_{2}(\alpha) < 2$ and $H^\mathrm{a}_{m}(\alpha) = m$ for all $m\le k$.
\item[\textup{(b)}]For all  $k$ there exists an infinite sequence $\alpha$ such that $H^\mathrm{a}_{2}(\alpha) < 2$ and $H^\mathrm{na}_{m}(\alpha) = m$ for all $m\le k$.
\end{description}
\end{theorem}

\begin{proof}
(a)~Consider all $k$-bit strings. It is easy to arrange them in some order $B_0,\ldots$ such that the last bit of $B_i$ is the same as the first bit of $B_{i+1}$, for all $i$, and the last bit of the last block is the same as the first bit of the first block. For example, consider (for every $x\in\{0, 1\}^{k - 2}$) four $k$-bit strings $0x0, 0x1, 1x1, 1x0$ and concatenate these $2^{k - 2}$ quadruples in arbitrary order.

Then consider a periodic sequence $\alpha$ with period $B_0B_1\ldots B_{2^k-1}$. Obviously all aligned $k$-bit blocks appear with the same frequency in $\alpha$, so $H^\mathrm{a}_{k}(\alpha) = k$. However, for non-aligned bit blocks of length $2$ we have two cases: this pair can be either completely inside some $B_i$, or be on the boundary between blocks. The pairs of the first type are balanced (since we have all possible $k$-bit blocks), but the boundary pairs could be only $00$ or $11$ due to our construction. So the non-aligned frequency of these two blocks is $1/4+\Omega(1/k)$, and for two other blocks we have $1/4-\Omega(1/k)$, so $H^\mathrm{na}_{2}(\alpha)  < 2$.

The only problem is that in this construction we do not necessarily have that $H^\mathrm{a}_m(\alpha) = m$ for all $m < k$, only for $m=k$. But this is easy to fix. Note that $H^\mathrm{a}_k(\alpha) = k$ implies $H^\mathrm{a}_m(\alpha) = m$ if $m$ is a divisor of $k$. So we can just use the same construction with blocks of length $k!$ instead of $k$.

(b)~Now let us consider a sequence constructed in the same way, but let blocks $B_0,B_1,\ldots,B_{2^k-1}$ go in the lexicographical ordering. First let us note that all $k$-bit blocks have the same \emph{non-aligned} frequencies in the periodic sequence with period $B_0B_1\ldots B_{2^k-1}$. (For aligned $k$-blocks it was obvious, but the non-aligned case needs some proof.) Indeed, consider some $k$-bit string $U$; we need to show that it appears exactly $k$ times in the (looped) sequence $B_0B_1\ldots B_{2^k-1}$. In fact, it appears exactly once for each position modulo~$k$.\footnote{So the cyclic sequence $B_0B_1\ldots B_{2^k-1}$ forms a \emph{perfect necklace} in the sense of Alvarez, Becher, Ferrari and Yuhjtman~\cite{alvarez-perfect} who note that such a sequence can be constructed in the same way as de Bruijn sequences (as an Eulerian path in some graph). Unfortunately, for our purposes we need also to guarantee that aligned $2$-bit blocks do not have the same frequencies, and for that we use our specific perfect necklace.} For example, it appears once among the blocks $B_i$. Why the same is true for some other position $s$: the $k-s$ first bits of $U$ appear as a suffix of $B_{i - 1}$ and  the last $s$ bits of $U$ appear as a prefix of $B_{i}$? Note that $(k - s)$-bit suffixes of $B_0, B_1, B_2,\ldots$ form a cycle modulo $2^{k-s}$, so the first $k - s$ bits of $U$ uniquely determine the \emph{last} $k - s$ bits of $B_{i}$, whereas the first $s$ bits of $B_i$ are just written in the $s$-bit suffix~of~$U$.

 This implies that non-aligned frequencies for all $k$-bit blocks are the same. Therefore, they are the same also for $m$-bit blocks for all $m\le k$. This implies also that we may assume for the rest of the proof that $k$ is odd.

Now let us consider \emph{aligned} blocks of size $2$. We will show that aligned frequency of the block $10$ in the sequence $B_0B_1\ldots B_{2^k-1}$  is  $1/4 - \Omega(1/k)$. Since $k$ is odd (see above), when we cut our sequence into blocks of size $2$, there are ``border'' blocks that cross the boundaries between $B_i$ and $B_{i+1}$, and other non-border blocks. Each second boundary is crossed (between $B_0$ and $B_1$, then $B_2$ and $B_3$, and so on), so the border blocks \emph{all have the first bit $0$}. In particular, $10$ never appears on such positions. This create a imbalance of order $1/k$ for block $10$, and we should check that it is not compensated by non-boundary blocks. In the blocks $B_i$ with even $i$ we delete that last bit and cut the rest into bit pairs. After deleting the last bit we have all possible $(k-1)$-bit strings, so no imbalance arises here. In the blocks $B_i$ with odd $i$ we delete the first bit, and then cut the rest into bit pairs. In the last pair the last bit is $1$ (since $i$ is odd), so once again we never have $10$ here, as required (the other positions are balanced).
\end{proof}

\subsection{Finite-state dimension and Wall's theorem}
\label{subsec:fsd-wall}

Using the notion of finite-state dimension, one can generalize Wall's theorem, as noted by Doty, Lutz and Nandakumar~\cite{doty}]

\begin{theorem}[Doty, Lutz, Nandakumar]\label{th:doty}
The finite-state dimension of a real number does not change when the number is multiplied by a rational number or when a rational number is added.
\end{theorem}

\begin{proof}
To prove this result, Doty, Lutz and Nandakumar show that for every $k$ the block entropy rates for $k$-bit blocks in a binary representation of a real number do not change significantly when a real number is multiplied by an integer. This obviously implies the same for the division by an integer, and adding integers is trivial, so the finite-state dimension does not change when we multiply by rational numbers or add rational numbers. More precisely, they show that
\[
\left|H^\mathrm{a}_{k}(\alpha) - H^\mathrm{a}_{k}(M\cdot \alpha) \right| \le \log_2(M^2(s + 1))
\]
for every real $\alpha$ and every positive integer $M$, where $s$ is the number of ones in the binary expansion of $M$. This inequality implies that finite-state dimensions of $\alpha$ and $M\alpha$ are the same, since the bound does not depend on $k$ and, being divided by $k$, converges to $0$ as $k\to\infty$.  In fact, a much simpler argument provides a better bound:

\begin{lemma}\label{lem:doty-bound-simplified}
For any real $\alpha$ and any positive integer $M$:
$$\left|H^\mathrm{a}_{k}(\alpha) - H^\mathrm{a}_{k}(M\cdot \alpha) \right| \le \log_2M\ \text{ and }\  \left|H^\mathrm{na}_{k}(\alpha) - H^\mathrm{na}_{k}(M\cdot \alpha) \right| \le \log_2M. $$
\end{lemma}
\begin{proof}
Both inequalities (aligned and non-aligned versions) are proven in a similar way. Consider, for instance, the aligned case. Choose $i\in\{0, \ldots, N-1\}$ uniformly at random and let $X$ be the $i$th aligned $k$-bit block in $\alpha$. Define a random variable $Y$ for $M\cdot \alpha$ in a similar way.

As we have noted while proving Theorem~\ref{th:wall}, for each group of neighbor positions $(i,i+1,\ldots,i+k-1)$, the bits of $\alpha$ in these positions determine almost uniquely the bits of $M\alpha$ in the same positions, and vice versa. Here ``almost uniquely'' means that there are at most $M$ possibilities. Therefore
\[
H(X\cnd Y) \le \log_2M \ \text{ and }\ H(Y\cnd X) \le \log_2M.
\]
Since $H(X) \le H(X, Y)  = H(Y) + H(X\cnd Y)$ and $H(Y) \le H(X, Y) = H(X) + H(Y\cnd X)$, we have
$$|H(X) - H(Y)| \le \log_2 M.$$
\end{proof}

As we have said, Lemma~\ref{lem:doty-bound-simplified} immediately implies Theorem~\ref{th:doty}.
\end{proof}

The bounds provided by Lemma~\ref{lem:doty-bound-simplified} are sharp, as the following example shows. Note that
\[
1/3 = 0.(01), \qquad 1/9 = 0.(000111),
\]
(parentheses show the period of a periodic fraction), 
$H^\mathrm{na}_{6}(1/3) = \log_2 2, H^\mathrm{na}_{6}(1/9) = \log_2 6$, and
$H^\mathrm{a}_{2}(1/3) = \log_2  1,  H^\mathrm{a}_{2}(1/9) = \log_2 3.$

\subsection{Finite-state dimension and automatic complexity}\label{subsec:fsd-ac}

The characterization of normal sequences in terms of automatic Kolmogorov complexity can be extended to the case of arbitrary finite-state dimension.

\begin{theorem}\label{th:fsd-automatic}
Finite-state dimension of an arbitrary bit sequence $\alpha=a_0a_1a_2\ldots$ is equal to
\[
\inf_R \,\liminf_{n\to\infty} \frac{\KS_R(a_0a_1\ldots a_{n-1})}{n}
\]
\end{theorem}

Note that replacing $\KS_R$ by the standard Kolmogorov complexity, we get the definition of the effective Hausdorff dimension, and $\inf_R$ is no more needed, since there exists an optimal description mode.

\begin{proof}
This result is a generalization of Theorem~\ref{th:compress} and the proof follows the same scheme. We need to prove two inequalities. In one direction we assume that the finite-state dimension of $\alpha$ is small: $\dimFS(\alpha)$ is less than some $\tau$. Then we need to construct an automatic description mode $R$ such that
\(
\liminf_{n\to\infty}{\KS_R(a_0a_1\ldots a_{n-1})}/{n} < \tau.
\)
The basic idea: if for some $k$ the distribution on aligned $k$-blocks has small entropy, then the corresponding Shannon -- Fano code has small average coding length and therefore provides good compression ratio when used as a description mode. However, there are two problems with this plan. First, for different prefixes the distributions on $k$-blocks, while having small entropy, could be different. Still we need to construct \emph{one} description mode that provides good compression for infinitely many prefixes. Second, the Shannon -- Fano code does not reach the exact value of entropy, the average length may exceed entropy (though not much, at most by $1$).

We have already seen in the proof of Theorem~\ref{th:compress} how to deal with both problems. For the first one, we consider the sequence of distributions with small entropies, use compactness to choose a convergent subsequence, construct the Shannon -- Fano code using the limit distribution, and modify it to cover blocks with zero probabilities. For the second, we note that the overhead $1$ (or $2$ due to the modifications of the code) is divided by the length of the block, so we may make the difference per bit arbitrarily small by considering large blocks. In the proof of Theorem~\ref{th:compress} we doubled the length of the block for this. Now we may do the same or use a similar argument implicitly by using Theorem~\ref{th:fsd} that allows us to start with blocks of arbitrarily large length.

Proving the inequality in the other direction, we assume that  
\[
\dimFS(\alpha)=\lim_k (H_k^\mathrm{a}(\alpha)/k)
\]
is high. This means that $H_k^\mathrm{a}(\alpha)/k$ is high for all sufficiently large~$k$. In fact, it is high for all $k$, since $\dimFS$ can be defined as the infimum of the same sequence, but this is not important for us now. We fix some automatic description mode $R$. We have to prove the lower bound for the automatic complexity $\KS_R$ for all (long enough) prefixes of $\alpha$. For that we cut a prefix into aligned blocks of large size $k$. We use superadditivity of $\KS_R$ and note that the $\KS_R$-complexity of the entire prefix is at least the sum of the $\KS_R$-complexities of the blocks. 

The rest of the proof is easy to explain if we use the prefix version of Kolmogorov complexity (see~\cite{kolmbook} for its definition and properties). It is close to the standard (plain) Kolmogorov complexity $\KS(\cdot)$ and therefore can be used as a lower bound for $\KS_R$ (with logarithmic precision). On the other hand, the prefix complexity by definition provides a prefix-free code for blocks, so the average prefix complexity (per block) has entropy as a lower bound. It is easy to finish this argument by noting that (a)~the constant in the inequality connecting $\KS_R$ and $\KS$ depends only on $R$, but not on the block size, and (b)~the difference between plain and prefix complexities is $O(\log k)$ for blocks of size~$k$ and does not matter for large $k$. However, we prefer to reformulate this argument to avoid using prefix complexity (see below). This is useful for readers who are not familiar with algorithmic information theory, and also will allow us to give a complexity-free characterization of finite-state dimension.

Now we give the details. 

\emph{First part}. Here we use the same Lemma~\ref{lem:shannon-fano} that was used in the proof of Theorem~\ref{th:compress}. Consider some sequence $\alpha$ whose finite-state dimension (in aligned version) is smaller than some threshold $\tau$. Since $\dimFS^{\mathrm{a}}(\alpha)=\lim_k H_k^{\mathrm{a}}(\alpha)/k$,  for all sufficiently large $k$ we have $H_k^{\mathrm{a}}(\alpha)<k\tau$.  Fix one of these values of~$k$. 

By definition $H_k^{\mathrm{a}}(\alpha)$ is the $\liminf$ of entropies of random aligned $k$-blocks in the growing prefixes of $\alpha$. For every $N$ that is a multiple of $k$ we consider the distribution $Q_N$ on $\mathbb{B}^k$ for the random aligned $k$-block in the $N$-bit prefix. The set of all distributions on $\mathbb{B}^k$ is compact. Therefore we may choose an increasing sequence of lengths $N_0, N_1,\ldots$ (all being multiples of $k$) such that the corresponding distributions $Q_{N_i}$ converge to some distribution $P$ on $\mathbb{B}^k$, and $H(Q_{N_i})<k\tau$ for all $i$. The entropy is a continuous function on the set of all distributions, therefore $H(P)\le k\tau$.

Assume first that all values of~$P$ are positive (no zeros).  Then we apply Lemma~\ref{lem:shannon-fano} to the distribution $P$ and get some automatic description mode $R$ with an upper bound for the $\KS_R$-complexity. We use this upper bound (divided by $N_i$) for prefixes of length $N_i$ and get (dividing by $N_i$) the bound
\[
\KS_R(a_0a_1\ldots a_{N_i-1})/{N_i} \le \frac{1}{k}\left(\sum_B Q_{N_i}(B)\log\frac{1}{P(B)}+1\right).
\]
Since $Q_{N_i}$ converge to $P$, the sum converges to $H(P)\le k\tau$, and we get an upper bound for the $\liminf$:
\[
\liminf_{N\to\infty}{\KS_R(a_0a_1\ldots a_{N-1})}/{N} \le \tau+1/k.
\]
This can be done for all sufficiently large $k$. For each $k$ we get some automatic description mode $R$ depending on $k$. Therefore, the infimum taken over all description modes is at most $\tau$, and this is what we need.

If some values of $P$ are zeros (some blocks have zero probability in the limit distribution), we cannot use the code provided by Lemma~\ref{lem:shannon-fano} since it does not provide codewords for blocks that have zero probability in $P$. As we noted in the proof of Theorem~\ref{th:compress}, one may change the code provided by the Lemma, adding leading $0$ to all codewords, and then use codewords starting from $1$ to encode ``bad'' blocks that have zero probability. Then all blocks have codes of finite length, the constant $1$ in the Lemma is replaced by $2$, and we can proceed as before. The exact lengths of codewords for bad blocks do not matter, since the limit frequencies of bad blocks are zeros.\footnote{The other way to deal with bad blocks is to use some $P'$ that is close enough to $P$ and has no zero probabilities; the overhead depends on the Kullback -- Leibler distance between $P$ and $P'$ and can be made arbitrarily small.}

The first part is proven.

\emph{Second part}. Assume that  $\dimFS(\alpha)=\lim_k (H_k^\mathrm{a}(\alpha)/k)>\tau$ for some $\tau$. Then $H_k^\mathrm{a}(\alpha)>k\tau$ for all sufficiently large $k$.  Fix  some $k$ with this property. Then the $\liminf$ of entropies of the $k$-blocks in growing prefixes exceeds $k\tau$. So, for this $k$ and for all sufficiently long prefixes the entropy of the corresponding distribution on aligned $k$-blocks is greater than~$k\tau$. Fix some automatic description mode~$R$. We will prove that for large enough prefixes $a_0a_1\ldots a_{n-1}$ the $\KS_R$-complexity per bit (i.e., divided by~$n$) is large, namely, exceeds $\tau-O(k/n)-O(\log k/k)$, where hidden constants do not depend on $n$ and $k$. Taking $\liminf$ when $n\to\infty$, we get rid of $O(k/n)$ and concluder that $\liminf_n (\KS_R(a_0\ldots a_{n-1})/n) \ge \tau - O(\log k/k)$. This lower bound works for arbitrarily large~$k$, therefore $\liminf_n (\KS_R(a_0\ldots a_{n-1})/n) \ge \tau$. In this argument $R$ is an arbitrary automatic description mode, so 
\(
\inf_R\liminf_n \KS_R(a_0\ldots a_{n-1})/n \ge \tau,
\)
as required.

To get the lower bound for $\KS_R(a_0a_1\ldots a_{n-1})$, we use that $\KS_R$ is superadditive (see Definition~\ref{def:superadditive}) and satisfies some calibration properties saying that there are not too many strings with small values of $\KS_R$. Namely, there are at most $O(2^m)$ strings $x$ with $\KS_R(x)\le m$, since all of them have descriptions of length at most $m$, there is at most $O(2^m)$ descriptions of this kind, and each of them serves $O(1)$ strings by definition of an automatic description mode. We use a bit weaker calibration condition in the following lemma since it will be useful later.

\begin{lemma}\label{lem:superadditive}
Let $F(x)$ be a superadditive non-negative real-valued function on strings. Assume that $F$ satisfies the following calibration condition:
\[
\sum_{|x|=s} 2^{-F(x)}\le \poly(s)
\]
for every length $s$. Then for every string $x$ that is a concatenation of $k$-bit blocks for some $k$, if $Q$ is the distribution on $\mathbb{B}^k$ that corresponds to the frequencies of these blocks in $x$, we have
\[
F(x) \ge \frac{|x|}{k}\left(H(Q)-O(\log k)\right).
\]
\end{lemma}
In the calibration condition $\poly(s)$ denotes some polynomial in $s$, i.e., we assume the polynomial growth of the sum of $2^{-F(x)}$ taken over all $s$-bit strings $x$, as a function of~$s$.

\begin{proof}
Assume that $x$ consists of $m=|x|/k$ bit blocks of length $k$, so $x=B_0B_1\ldots B_{m-1}$. Then, due to the superadditivity of $F$,
\[
F(x)=F(B_0B_1\ldots B_{m-1})\ge F(B_0)+\ldots+F(B_{m-1}),
\]
and we need to get a lower bound for the sum in the right hand side. For that, note that for an integer-valued function $F'(x)=\lfloor F(x)+c\log |x|\rfloor$ we have
\[
\sum_{|x|=s} 2^{-F'(x)}\le 1,
\]
for all $s$, if $c$ is a large enough constant. Indeed,  the $c\log |x|$ additive term in the exponent compensates for $O(\poly(s))$-factor and rounding. Using this property for $s=k$, we conclude that there exists a prefix code for $k$-bit strings where the codeword for a string $x$ has length $F'(x)$. The average length of this code for $k$-blocks distributed according to $Q$ is at least $H(Q)$, so we have
\[
\frac{F'(B_0)+\ldots+F'(B_{m-1})}{m} \ge H(Q).
\]
Therefore,
\[
F(x)\ge F(B_0)+\!\ldots\!+F(B_{m-1})\ge F'(B_0)+\!\ldots+\!F'(B_{m-1})-mO(\log k) \ge m(H(Q)-O(\log k)),
\]
as claimed.
\end{proof}

We will apply this lemma to $F=\KS_R$, so we need to show that $\KS_R$ satisfies the calibration condition. As we mentioned, there is at most $O(2^m)$ strings $x$ with $\KS_R(x)\le m$, and this is enough:

\begin{lemma}\label{lem:two-calibre}
Assume that $F$ is a non-negative function on strings, and for every integer $m\ge 0$ there is at most $O(2^m)$ strings $x$ such that $F(x)\le m$. Then $F$ satisfies the calibration condition of Lemma~\ref{lem:superadditive}. 
\end{lemma}

\begin{proof}
Consider the sum $\sum_{|x|=s} 2^{-F(x)}$ for some $s$. It contains $2^s$ terms. For some of these terms $F(x)>s$, so each of them is less than $2^{-s}$ and the sum is at most $1$.  All other terms we classify into $s+1$ groups according to the value of $\lfloor F(x)\rfloor$. The group where $\lfloor F(x)\rfloor = i$ contains $O(2^i)$ terms and each is at most $2^{-i}$, so in total we get $O(s)\le \poly(s)$, as required.
\end{proof}

\begin{remark}
Due to $\poly(s)$ factor in the calibration condition,  it remains valid if we subtract $O(\log |x|)$-term from $F(x)$. Therefore, the calibration condition is true for all versions of Kolmogorov complexity: they differ by $O(\log n)$ for strings of length $n$, and for the plain complexity the condition of Lemma~\ref{lem:two-calibre} holds.\footnote{One can note also that the calibration condition is obviously true for prefix complexity, since the sum of $2^{-\KP(x)}$ over all $x$ (of any length) is at most~$1$; the same argument works for monotone complexity, since the set of all strings of a given length is prefix-free.}
\end{remark}

Now we may apply Lemma~\ref{lem:superadditive} to get a lower bound for $\KS_R(a_0a_1\ldots a_{n-1})$ for an arbitrary prefix $a_0\ldots a_{n-1}$ of $\alpha$. Take an arbitrary $k$ and let $m=\lfloor |x|/k\rfloor$. We split the prefix $a_0a_1\ldots a_{n-1}$ into $m$ blocks of length $k$ (deleting the last incomplete block that can only increase $\KS_R$) and use the bound provided by the lemma. Dividing by $n$, we see that 
\[
\frac{\KS_R(a_0a_1\ldots a_{n-1})}{n}\ge \tfrac{m}{n}(H(Q) - O(\log k)) \ge \tfrac{m}{n}(k\tau - O(\log k)) \ge \tau - \frac{O(k)}{n} - \frac{O(\log k)}{k}.
\]
for sufficiently large $n$. Here $Q$ is the distribution on $k$-blocks; its entropy is at least $\tau k$ for large enough $n$. The last step is valid since $n=k(m+O(1))$. So we get the desired inequality, and this finishes the proof of Theorem~\ref{th:fsd-automatic}.
\end{proof}

\begin{remark}\label{rem:alternative}
In fact, our proof of Theorem~\ref{th:fsd-automatic} gives a bit more that we claimed. Namely, we can prove the inequalities between the quantities used in the two definitions in the strongest possible form. We can show that
\begin{itemize}
\item If the aligned block entropy is small for some $k$ and for infinitely many prefixes, then infinitely many prefixes have small automatic complexity. To prove this, we first use the trick used to prove Theorem~\ref{th:compress} and not that the block entropy is also small for $2k$-bit blocks, $4k$-bit blocks etc. Then we use the limit distribution for $k$-bit blocks (or $2k$-bit blocks, or $4k$-bit blocks) to construct a code that provides a good compression ratio. Therefore, $\inf_R\liminf_n (\KS_R(a_0\ldots a_{n-1})/n) \le H_k^\mathrm{a}(\alpha)/k$ for every $k$ (and every sequence $\alpha=a_0a_1\ldots$).
\item If the block entropy is large for some $k$ and for all sufficiently long prefixes, then all long prefixes have large automatic  complexity. For that we split the sequence into $k$-bit blocks, use superadditivity and provide a bound for compression ratio (with error $O(\log k/k)$). So it is enough to have infinitely many $k$ with large $H^\mathrm{a}_k(\alpha)/k$.
\end{itemize}
These two arguments show that
\[ 
\limsup_k (H^\mathrm{a}_k(\alpha)/k) \le
\inf_R \,\liminf_n (\KS_R(a_0\ldots a_{n-1})/n)
\le \inf_k (H^\mathrm{a}_k(\alpha)/k),
\]
so we conclude that 
\[
\lim_k (H^\mathrm{a}_k(\alpha)/k) =
\inf_k (H^\mathrm{a}_k(\alpha)/k)
\]
(and the limit exists) without using Theorem~\ref{th:fsd}. Moreover, we can adapt the proof of Theorem~\ref{th:fsd-automatic} for non-aligned blocks to prove a similar equality for the non-aligned case, thus deriving the full statement of Theorem~\ref{th:fsd} without using information-theoretic arguments like Shearer-type inequality. Let us sketch this argument.

In one direction we assume that the non-aligned block entropy for some block size $k$ is small, and want to show that infinitely many prefixes are compressible enough. Note that the distribution for non-aligned $k$-blocks is the average of $k$ distributions that correspond to aligned blocks in the original sequence $\alpha$, then in $\alpha$ without the first bit, then in $\alpha$ without two first two bits, etc. The average of the entropies of these distributions is smaller than the entropy of the average distribution (convexity of entropy; we discussed it for the case of two distributions), so one of these $k$ sequences has compressible prefixes. The deleted bits then can be added back without changing much the automatic complexity, so the original sequence $\alpha$ is also compressible.

In the other direction things are a bit more complicated. We get a lower bound for the automatic complexity by splitting the sequence into $k$-bit blocks, but this lower bound involves the entropy of the aligned distribution. Of course, we can shift the boundaries modulo $k$, and get another lower bound for the same automatic complexity that involves another aligned distribution (for the sequence without first $0$, $1$,\ldots, $k-1$ bits). Averaging this bounds, we get a boundary that involves the average of the entropies of these $k$ distributions. (We can take maximum instead of the average, but we will not need this.) The problem, however, is that this average may be smaller than the entropy of the non-aligned distribution (that is the average of $k$ aligned distributions). More precisely, this average entropy is the \emph{conditional} entropy of the non-aligned distribution when the condition is the position of the block modulo $k$. It remains to note that the difference between unconditional and conditional entropy is bounded by the entropy of the condition, i.e., $\log k$. Since we study the entropy per bit and divide the entropy by $k$, this difference does not matter ($\log k/k\to 0$).

\end{remark}

\subsection[Machine-independent characterization of normal sequences and finite-state dimension]{Machine-independent characterization of normal sequences and finite-state dimension}\label{subsec:stateless}

In fact we have proven the following characterization of finite-state dimension that does not mention explicitly finite automata or entropies.

\begin{theorem}\label{th:stateless}
Let $\alpha=a_0a_1\ldots$ be an infinite bit sequence. Then the finite-state dimension of $\alpha$ is the infimum
\[
\inf_F \,\liminf_n \frac{F(a_0a_1\ldots a_{n-1})}{n}
\]
taken over all superadditive functions $F$ that are upper bounds of Kolmogorov complexity with logarithmic precision, i.e., $\KS(x)\le F(x)+O(\log |x|)$ for all $x$.
\end{theorem}

\begin{proof}
Recall Theorem~\ref{th:fsd-automatic} and its proof. We need to make two more remarks:

First, we may use $\KS_R$ (or $\KA_R$ defined in the next section) as $F$, and this shows that the $\inf_F$ in question does not exceed the finite-state dimension. 

The inequality in the other direction is already proven since in the proof of Theorem~\ref{th:fsd-automatic} we used only the superadditivity and the calibration property, and have noted that the calibration property is true for every $F$ that is an upper bound for the Kolmogorov complexity with logarithmic precision.
\end{proof}

One can say that this result explains the intuitive meaning of finite-state dimension: it measures the compressibility of prefixes of $\alpha$ if only ``local'' compression/decompression methods are allowed for which any splitting the uncompressed sequence induces a splitting of its compressed version.

\begin{remark}
This theorem is quite robust:
\begin{itemize}
\item We can replace the term $O(\log |x|)$ by $O(1)$, since the function $\KS_R$ used in the proof does not exceed $\KS(x)+O(1)$. 
\item We can also replace the term $O(\log|x|)$ by $o(|x|)$, since it is enough for the proof of the lower bound for $\liminf F(a_0\ldots a_{n-1})/n$.
\end{itemize}
\end{remark}

One can delete all references to Kolmogorov complexity replacing them by the calibration condition.

\begin{theorem}\label{th:stateless1}
Let $\alpha=a_0a_1\ldots$ be an infinite bit sequence. Then the finite-state dimension of $\alpha$ is the infimum
\[
\inf_F \,\liminf_n \frac{F(a_0a_1\ldots a_{n-1})}{n}
\]
taken over all superadditive functions $F$  such that 
\[
\sum_{|x|=k} 2^{-F(x)}=O(\poly(k))
\]
\end{theorem}

\begin{proof}
No new argument is needed, since only this calibration condition was used in the proof, and the function $\KS_R$ satisfies this calibration condition (as we have shown).
\end{proof}

\begin{remark}
We can replace the calibration condition by the other one (that is satisfied by the plain Kolmogorov complexity function): the number of strings $x$ such that $F(x)<m$, is $O(2^m)$. Indeed, the function $\KS_R$ satisfies this condition. On the other hand, we have seen that it implies the condition used in Theorem~\ref{th:stateless} (Lemma~\ref{lem:two-calibre}).
\end{remark}

Later (Section~\ref{subsec:fsm}, remark after Theorem~\ref{th:fsd-apriori}) we will see another calibration condition that can be used in the theorem: $\sum_{x\in P} 2^{-F(x)}\le c$ for some $c$ and for every prefix-free set $P$. It obviously implies the condition used in Theorem~\ref{th:stateless}, so we need only to provide a proof in the other direction, i.e., show a superadditive function that satisfies this condition and can be used in the proof instead of $\KS_R$. This function will be a finite-state version of a priori probability (maximal continuous semimeasure in the Kolmogorov complexity theory, see~\cite{kolmbook}).

The characterization of finite-state dimension in terms of automatic complexity allows us to extend the sufficient condition for normality (Section~\ref{subsec:champernowne}, Theorem~\ref{th:normality-condition}) and get the following lower bound for the finite-state dimension.

\begin{theorem}\label{th:lower-bound-dimension}
Let $x_1,x_2,x_3,\ldots$ be a sequence of non-empty binary strings. Let $L_n$ be a rational number that is the average length of $x_1,\ldots,x_n$, i.e., $L_n= (|x_1|+\ldots+|x_n|)/n$. Let $C_n$ be their average Kolmogorov complexity, i.e., $C_n=(\KS(x_1)+\ldots+\KS(x_n))/n$. Assume that $|x_n|=o(|x_1|+\ldots+|x_{n-1}|)$ and $L_n\to\infty$ as $n\to\infty$. Then the finite-state dimension of the bit sequence $\varkappa=x_1x_2x_3\ldots$ is at least $\liminf_n (C_n/L_n)$.
\end{theorem}

\begin{proof}
Using the characterization of the finite-state dimension in terms of automatic complexity, we need to show that for every automatic description mode $R$ the liminf of $\KS_R(u)/|u|$, where $u$ is a prefix of $\varkappa$, is at least $\liminf_n C_n/L_n$. If $u$ ends on the block boundary, i.e., if $u=x_1\ldots x_n$ for some $n$, then
\[
F(u)=F(x_1\ldots x_n)\ge F(x_1)+\ldots+F(x_n)\ge \KS(x_1)+\ldots+\KS(x_n)-O(n)=nC_n-O(n),
\]
since $\KS(x)\le \KS_R(x)+O(1)$. At the same time $|u|=nL_n$, since $L_n$ is the average length of the first $n$ blocks, so
\[
\KS_R(u)/|u|\ge C_n/L_n - O(1)/L_n.
\]
That gives the desired bound for prefixes that end on the block boundaries.

Now we should consider $u$ that do not end on the block boundary. We can delete the last incomplete block and get a slightly shorter $u'$. For this $u'$ we use the same bound as before, and due to the superadditivity it works as a bound for $u$. However, we have $|u|$ in the denominator, not $|u'|$. This does not change the $\liminf$, since we assume that $|x_n|=o(|x_1|+\ldots+|x_{n-1}|)$, so the length of the incomplete block is negligible compared to the total length of previous complete blocks, and the correction factor converges to~$1$.  Theorem~\ref{th:lower-bound-dimension} is proven.
\end{proof}

\subsection{Finite-state martingales and automatic a priori complexity}\label{subsec:fsm}

The original definition of the finite-state dimension~\cite{dai-fsd} was given in terms of games (or martingales corresponding to games). In this section we review this definition and show that it is equivalent to the definitions given above. This equivalence was proven by Bourke, Hitchcock and Vinodchandran~\cite{bourke}; we provide a simple alternative proof based on a finite-state version of a priori probability. But let us first say a few general words about the game approach to randomness that goes back to Ville (see his book~\cite{ville1939}; see~\cite{bss} for more historic details).


The game approach to randomness is based on the following idea: a bit sequence is not random if we can become infinitely rich playing against this sequence. The game is as follows: before seeing the next bit of the sequence, we have some amount of money $m$ and split it into two parts $m=m_0+m_1$, making two bets (on zero and one). Then the next bit is shown, the wrong bet is lost and the correct bet is doubled.\footnote{One may wish to keep some part of the capital not using it for bets, but the same result can be achieved by betting half of it on zero and half of it on one.} So our capital after seeing the next bit $b$ is $2m_b$. The strategy in such a game is a function saying how we should split our capital after seeing a prefix of the sequence we are playing with. The usual way to describe the strategy is to provide the corresponding martingale, a non-negative real-valued function $m(x)$ that says what is our capital after playing with prefix $x$. It is easy to see that the rules of the game described above mean that
\[
m(x) =\frac{m(x0)+m(x1)}{2} \eqno(*)
\]
for every string $x$. 

\begin{definition}
A \emph{martingale} is a function $m$ with non-negative values defined on all binary strings that satisfies the equality $(*)$ for all $x$. 
\end{definition}

A more general notion of martingale is used in the probability theory, but for our purposes this special case is enough. After playing with prefix $x$, we split the capital $m(x)$ and make bets $m(x0)/2$ and $m(x1)/2$; the correct bet is doubled and our capital becomes $m(x0)$ or $m(x1)$.
 
\begin{definition}
A martingale \emph{wins} on a binary sequence $\alpha$ if it is not bounded on the prefixes of~$\alpha$.
\end{definition}

Martingales are in one-to-one correspondence with measures on the Cantor space. A measure on the Cantor space is determined by its values on \emph{intervals} $[x]$; here $[x]$ is an interval that contains all sequences that have prefix $x$. For a measure $\mu$ we have $\mu([x])=\mu([x0])+\mu([x1])$. If the the measure of the entire space is $1$, it is called a \emph{probability distribution}. The uniform Lebesgue measure $\lambda$ on the Cantor space is defined as $\lambda([x])=2^{-|x|}$ and corresponds to independent fair coin tosses. The following statement follows directly from the definitions:

\begin{proposition}
If $\mu$ is some measure on the Cantor space of bit sequences, then 
\[
m(x)=\frac{\mu([x])}{\lambda([x])}
\]
is a martingale. Every martingale corresponds to some measure in this way. Martingales that equal $1$ on the empty string correspond to probability distributions.
\end{proposition}

The conditional probabilities $\mu([x0])/\mu([x])$ and $\mu([x1])/\mu([x])$ are the fractions of capital that are used to bet on $0$ and $1$ respectively, after seeing the prefix~$x$.


There is a characterization of Martin-L\"of random sequences in terms of martingales~\cite{schnorr1971}: a sequence is Martin-L\"of random if and only if no lower semicomputable martingale wins against this sequence. Lower semicomputability means that there is an algorithm that, given a string $x$, produces an increasing computable sequence of rational numbers that converges to the martingale value $m(x)$. A ``scaled-down'' version of this result~\cite{schnorr-stimm} says that a sequence is normal if and only if no finite-state martingale wins against it.  Informally speaking, finite-state martingales correspond to strategies with finite memory, i.e., the strategies that have finite number of states, each state determines the proportion of bets (a pair of rational numbers whose sum is $1$), and the next state is determined by a previous state and the bit seen). We will reprove this result in Section~\ref{subsec:agafonov} (Theorem~\ref{th:s-and-s}).

Martingales can also be used to define effective Hausdorff dimension and (in the case of finite-state martingales) finite-state dimension~\cite{lutz-dimension,lutz-individual}. The dimension of a bit sequence determines how fast a martingale can grow on the prefixes of this sequence: dimension is the infimum of $s$ such that $m(x)/2^{(1-s)|x|}$ is unbounded for some martingale $m$. The exponential growth of martingales was first considered by Schnorr~\cite[Chapter 17]{schnorr1971}; much later these ideas were rediscovered and developed by defining effective Hausdorff dimension~\cite{lutz-dimension,lutz-individual,mayordomo}.

Technically it is convenient to introduce the notion of \emph{$s$-gale} for $s\in [0,1]$.

\begin{definition}
Let $s\in [0,1]$. An $s$-gale is a function $m(x)$ on binary strings with non-negative real values such that
\[
m(x) =\frac{m(x0)+m(x1)}{2^s} 
\]
\end{definition}

This definition introduces a ``tax'': after each game the capital is multiplied by factor $2^{s-1}$. For $s=1$ we have no tax: $1$-gales are just martingales. For $s=0$ we have tax rate $50\%$ (half of the capital is taken away after each game). In the latter case ($s=0$) we cannot win: even if we guess all the bits correctly and make the corresponding bets, our capital will only remain unchanged.

It is easy to see that we may equivalently define $s$-gales as functions of type
\[
m(x)=\frac{\mu([x])}{2^{-s|x|}}
\]
where $\mu$ is some measure.

The effective Hausdorff dimension of a bit sequence $\alpha=a_0a_1\ldots$ can be defined as the infimum of the values of $s$ such that some lower semicomputable $s$-gale wins on~$\alpha$. One can also equivalently define it in terms of Kolmogorov complexity as 
\[
\liminf \KS(a_0a_1\ldots a_{n-1})/n
\]
(see~\cite{lutz-dimension,lutz-individual,mayordomo,staiger2005} or~\cite[Section~13.3]{downey-hirschfeldt} and \cite[Sections 5.8 and 9.10]{kolmbook} for a survey). This subject has a long history (see the discussion in \cite[footnote on p.~598]{downey-hirschfeldt}); we do not go into details since we are interested only in the parallel theory developed in~\cite{dai-fsd} for the finite-state case. It used finite-state gales to define the finite-state dimension. Then in~\cite{bourke} the equivalence between the definitions of finite-state dimension in terms of gales and entropy rates was proven. We used the reverse order: the definition of finite-state dimension was given in terms of the entropy rates, and now we are going to give a simple proof of its equivalence to martingale definition.


First let us give the exact definitions. Assume that a finite set of nodes (states) is given; one of them is called an initial state. For each node (state) there are two outgoing edges labeled $(0,p_0)$ and $(1,p_1)$, where $p_0$ and $p_1$ are non-negative rational numbers and $p_0+p_1=1$. This labeled graph (together with the initial state) determines a probabilistic process: it starts in the initial state and then changes the state in a natural way: an outgoing edge is selected with probability written on it, i.e., the probability to choose an edge is the second component of a pair on that edge. The corresponding bit (the first component of the pair) is sent to the output. We get a measure on the Cantor space.

\begin{definition}
Measures of this type are called \emph{finite-state measures}. If $\mu$ is a finite-state measure on the Cantor space, the ratio $\mu([x])/2^{-s|x|}$ is called a \emph{finite-state $s$-gale}. A finite-state $1$-gale is called a \emph{finite-state martingale}.
\end{definition}

Schnorr and Stimm~\cite{schnorr-stimm} introduced finite-state martingales under the name (in German) ``Verm\"ogenfunktionen erzeugt von endlichen Automaten''. Dai, Lathrop, Lutz and Mayordomo in their paper~\cite{dai-fsd} use the name ``$1$-account finite-state $s$-gale'' for the notion we consider. It is easy to see that finite-state martingales correspond to the gambling strategies described above (the gambler's decision how to split the capital between two bets is computed by a finite automaton: nodes are states of this automaton, the rational numbers on the outgoing edges determine the bets, and the endpoints of the edges are next states for two possible values of the next bit).

\begin{theorem}\label{th:fsg}\leavevmode

\textup{\textbf{1}}. \textup(Schnorr and Stimm,  \textup{\cite[Satz 4.1]{schnorr-stimm}}\textup) A sequence is normal if and only if no finite-state martingale wins against it.

\textup{\textbf{2}}. \textup(Dai, Lathrop, Lutz, Mayordomo, \textup{\cite{dai-fsd}}\footnote{As we have said, in~\cite{dai-fsd} this property was used as a definition of finite-state dimension, so they formulated the result as the characterization of finite-state dimension in terms of aligned entropy rate.}\textup) The finite-state dimension of $\alpha$ is the infimum of $s\in[0,1]$ such that there exists a finite-state $s$-gale that wins against $\alpha$; if there is no such $s$, the finite-state dimension is~$1$.
\end{theorem}

We start by proving the second statement. It gives immediately one implication in the first statement: if no martingale wins on a sequence, then no $s$-gale (for $s\in[0,1]$) can win on it, the finite-state dimension is $1$ and the sequence is normal. We postpone the proof of the reverse implication in the first part, since it uses some additional technique that goes back to Agafonov~\cite{agafonov}.  See below Theorem~\ref{th:s-and-s} for this implication.


To prove Theorem~\ref{th:fsg}, we follow the proof of Theorem~\ref{th:fsd-automatic} with some changes. Namely, we replace the notion of automatic complexity (as defined in Section~\ref{sec:complexity}) by a similar notion that resembles a priori complexity (logarithm of the maximal continuous semimeasure as defined in the algorithmic information theory, see~\cite{kolmbook} for details). 

\begin{definition}
Let $R$ be a finite-state probabilistic process (i.e., a labeled graph of the described type). Consider state $i$ as its initial state, and let $\rho_{R,i}$ be the corresponding measure on the Cantor space.  Its logarithm can be considered as a complexity measure, and we define
\[
\KA_{R,i} (x) = - \log \rho_{R,i}([x])
\]
for a binary string $x$. For a fixed graph $R$, we consider all its nodes as initial state, take the minimum over all nodes $i$ and let
\[
\KA_R (x) = \min_i \KA_{R,i}(x);
\]
\end{definition}

In other words, we consider the maximal probability over all initial states, and its negative logarithm. This is technically important to get a superadditive function; for similar reasons we did not fix an initial state when defining the automatic complexity.

The following result, analogous to Theorem~\ref{th:fsd-automatic}, is essentially a reformulation of the second part of Theorem~\ref{th:fsg}.

\begin{theorem}\label{th:fsd-apriori}
Finite-state dimension of an arbitrary bit sequence $\alpha=a_0a_1a_2\ldots$ is equal to
\[
\inf_R \,\liminf_{n\to\infty} \frac{\KA_R(a_0a_1\ldots a_{n-1})}{n}
\]
\end{theorem}

Before proving this result, let us show that the second part of Theorem~\ref{th:fsg} follows from it. For that we need to show two things:
\begin{itemize}
\item if $s>\liminf_{n\to\infty} (\KA_R(a_0a_1\ldots a_{n-1})/n)$ for some $R$, then there is an $s$-gale that wins against $\alpha$;
\item if there is an $s$-gale that wins against $\alpha$, then $s \ge\liminf_{n\to\infty} (\KA_R(a_0a_1\ldots a_{n-1})/n)$ for some $R$.
\end{itemize}
In both case we consider a probabilistic process $R$ that corresponds to the $s$-gale.  Note first that the inequality $s>\KA_R(a_0a_1\ldots a_{n-1})/n$ by definition means that 
\[
\rho_i ([a_0a_1\ldots a_{n-1}])> 2^{-sn}
\]
for some initial state $i$. So if this happens for some $R$ and infinitely many $n$ (as it should if $s$ exceeds $\liminf$), then we can select $i$ that appears infinitely often, and for that $i$ the corresponding $s$-gale exceeds $1$ infinitely often. It does not mean winning according to our definition, but we can slightly decrease $s$ first (in such a way that it still exceeds the $\liminf$) and apply the argument to the smaller gale. Then we change $s$ back and convert a sequence that exceeds $1$ infinitely often into an unbounded sequence.

On the other hand, if some $s$-gale wins against $\alpha$, then the corresponding distribution infinitely often exceeds $2^{-sn}$ for $n$-bit prefixes of $\alpha$, and $\KA_R(a_0a_1\ldots a_{n-1})$ is smaller than $sn$ for infinitely many prefixes. This finishes the derivation of the second part of Theorem~\ref{th:fsg} from Theorem~\ref{th:fsd-apriori}. Now let us prove Theorem~\ref{th:fsd-apriori}.

\begin{proof}[Proof of Theorem~\ref{th:fsd-apriori}]
We follow the scheme used for the proof of Theorem~\ref{th:fsd-automatic} with minimal changes. In the first part, we need to prove the version of Lemma~\ref{lem:shannon-fano} where $\KS_R$ is replaced by $\KA_R$. 

\begin{lemma}\label{lem:shannon-fano-2}
Let $k$ be some integer and let $P$ be a distribution on a set $\mathbb{B}^k$ of $k$-bit blocks. Then there exists a probabilistic process~$R$ such that for every string $x$ whose length is a multiple of $k$, we have
\[
\KA_R(x)\le \frac{|x|}{k}\left(\sum_B Q(B) \log\frac{1}{P(B)} + 1\right)
\]
where $Q$ is the distribution on $k$-bit blocks appearing when $x$ is split into blocks of size $k$.
\end{lemma}

\begin{proof}
We may consider the same prefix code from Shannon's theorem, and consider a probabilistic process that tosses a fair coin to choose the next bit of a growing string, and decodes this string with respect to the prefix code chosen (when the codeword's end is reached, we output the encoded string; this is done using several states where the next move is deterministic, i.e., has probability~$1$). The probability to get some string $x$ as the output is at least $2^{-m}$ if $m$ is the length of its description, so we get the same bound as in Lemma~\ref{lem:shannon-fano}.

In fact, a simpler argument that gives a better bound is possible. We do not really need the coding argument and corresponding $+1$ overhead. Instead, we may consider a finite-state probabilistic process that generates probability distribution $P$ on the consecutive blocks (different blocks are independent) and get directly the inequality
\[
\KA_R(x)\le \frac{|x|}{k}\left(\sum_B Q(B) \log\frac{1}{P(B)}\right)
\]
without the term ``+1''. Still one small problem remains: by definition, we want the transitional probabilities in the finite-state process to be rational numbers, so we need to replace $P$ by some its rational approximation, and again an (arbitrariraly small) additive term appears, instead of ``$+1$''. In this way we also make all probabilities in the approximate distribution strictly positive, and this is makes $\KA_R$ finite everywhere.
\end{proof}

The rest of the proof remains the same as for Theorem~\ref{th:compress}.

For the other direction, we apply Lemma~\ref{lem:superadditive} to the function $F(x)=\KA_R(x)$. This gives the desired result immediately, and it remains to show that function $\KA_R(x)$ has the required properties:

\begin{lemma}\label{lem:apriori-properties}
For every $R$ the function $\KA_R(x)$ is a superadditive function that satisfies the calibration condition.
\end{lemma}

\begin{proof}
We need to prove superadditivity: $\KA_R(uv)\ge \KA_R(u)+\KA_R(v)$. To get the required lower bound for $\KA_R(uv)$ we need to prove an upper bound for $\max_i\rho_i(uv)$ where maximum is taken over all nodes of $R$, i.e., to prove the same bound for every~$i$. But $\rho_i(uv)=\rho_i(u)\rho_j(v)$ where $j$ is the state where the process is after generating $u$, and this gives the required bound. Here it is essential that we take maximum of $\rho_i(x)$ over all $i$ (minimum of $\KA_{R,i}(x)$ over all $i$), this is a technical trick that makes $\KA_R$ superadditive.

The calibration condition (Lemma~\ref{lem:superadditive}) is easy to check: 
\[
2^{-\KA_R(x)}=\max_i 2^{-\KA_{R,i}(x)}\le \max_i \rho_{R,i}(x) \le \sum_i \rho_{R,i}(x).
\]
Now we compute the sum over all $x$ of given length $s$, and for each $i$ the sum of $\rho_{R,i}(x)$ over these $x$ is $1$, so the sum in the calibration condition does not exceed the number of states. So we get not only a polynomial in $s$ bound, but a constant bound.
\end{proof}

We could also prove this lemma saying that $\KA_R$ (for every $R$) is an upper bound for apriori Kolmogorov complexity (see~\cite[Section 5.3]{kolmbook}) and all the version of complexity differ only by a logarithmic term and satisfy the calibration condition.

This lemma finishes the proof of Theorem~\ref{th:fsd-apriori}, and therefore the second part of Theorem~\ref{th:fsg} is also proven. We will return to the first part after discussing Agafonov's result (see Theorem~\ref{th:s-and-s} below).
\end{proof}

\subsection{Agafonov's and Schnorr--Stimm's theorems}\label{subsec:agafonov}

In this section we derive another classical result about normal numbers, Agafonov's theorem~\cite{agafonov}, from the martingale characterization. 

Agafonov's result is motivated by the von Mises' approach to randomness (see, e.g.,~\cite[Chapter 9]{kolmbook} for a historic account). As von Mises had mentioned, a random sequence (he used German word \emph{Kollektiv}) should remain random after using a reasonable selection rule. More precisely, assume that there is some set $S$ of binary strings. This set determines a ``selection rule'' that selects a subsequence from every binary sequence $\alpha$. The selection works as follows: we observe a binary sequence $\alpha=a_0a_1a_2\ldots$ and select terms $a_n$ such that $a_0a_1\ldots a_{n-1}\in S$ (without reordering the selected terms). We get a subsequence; if an initial sequence is ``random'' (is plausible as an outcome of a fair coin tossing), said von Mises, this subsequence should also be random in the same sense. The Agafonov's theorem says that for regular (automatic) selection rules and normality as randomness this property is indeed true.

\begin{theorem}[Agafonov]\label{th:agafonov}
Let $\alpha=a_0a_1a_2\ldots$ be a normal sequence. Let $S$ be a regular \textup(=recognizable by a finite automaton\textup) set of binary strings. Consider a subsequence $\sigma$ made of terms $a_n$ such that $a_0a_1\ldots a_{n-1}\in S$, taken in the same order as in the original sequence. Then $\sigma$ is normal or finite.
\end{theorem}

\begin{proof}
We already know that a sequence is \emph{not} normal if and only if there is a finite-state $s$-gale for $s<1$ that wins against it. So we need to show that if some $s$-gale wins against a selected subsequence, then there is some other $s'$-gale for (may be, different) $s'<1$ that wins against the entire sequence. In terms of martingales: if some martingale wins exponentially fast playing against the selected (infinite) subsequence, then some other martingale wins exponentially fast against the entire sequence (may be, with different exponent).

In this language the idea of the proof is obvious. Assume that we have some strategy $\sigma$ that plays against the subsequence. Then we can play against the entire sequence as follows. We make the trivial bets ($0.5+0.5$) when we bet on the non-selected bits of the entire sequence, and use $\sigma$ to bet on the selected bits. Since the selection rule is defined by a finite automaton, it is easy to see that the new strategy also has finite memory, and the capital will be the same as in the game of $\sigma$ against the subsequence. The only problem is the rate: if the selected subsequence is very sparse, then exponential rate in the subsequence game is no more an exponential rate in the entire game. But the selected subsequence cannot be too rare, its density is separated from zero, as the following well known lemma\footnote{This lemma was implicitly used in the original proof of Agafonov~\cite{agafonov}; an explicit statement can be found, e.g., in~\cite[Lemma 7.5]{bch2}; note that in~\cite{becher-heiber} a simple special case is considered when the transition function is free of cycles of non-selecting states.} says:

\begin{lemma}\label{lem:density}
If the selected subsequence is infinite, then it has a positive density, i.e., the $\liminf$ of the density of the selected terms is positive.
\end{lemma}

\begin{proof}[Proof of the lemma]

Consider a deterministic finite automaton that recognizes the set $S$. We denote this automaton by the same letter $S$. Let $X$ be the set of states of $S$ that appear infinitely many times when $S$ is applied to $\alpha$. Starting from some moment, the automaton is in $X$, and $X$ is strongly connected (when speaking about strong connectivity, we ignore the labeling of the transition edges). Let us show that \emph{vertices in $X$ have no outgoing edges that leave $X$}. If these edges exist, let us construct a string $u$ that forces $S$ to leave $X$ when started from any vertex of $X$. This will lead to a contradiction: a normal sequence has infinitely many occurrences of $u$, and one of them appears when $S$ already is in $X$. 

How to construct this $u$? Take some $q\in X$ and construct a string $u_1$ that forces $S$ to leave $X$ when started from $q$. Such a string $u_1$ exists, since $X$ is strongly connected, so we can bring $S$ to any vertex and then use the letter that forces $S$ to leave $X$. Now consider some other vertex $q'\in X$. It may happen that $u_1$ already forces $S$ to leave $X$ when started from $q'$. If not and $S$ remains in $X$ (being in some vertex $v$), we can find some string $u_2$ that forces $S$ out of $X$ when started at $v$. Then the string $u_1u_2$ forces $S$ to leave $X$ when started in any of the vertices $q, q'$. Then we consider some other vertex $q''$ and append (if needed) some string $u_3$ that forces $S$ to leave $X$ when started at $q$, $q'$ or $q''$ (in the same way). Doing this for all vertices of $X$, we get the desired $u$ (and the desired contradiction).\footnote{One may use also a probabilistic argument: for every vertex $q\in X$ there is some string that forces $S$ to leave $X$ when started at $q$, so for a sufficiently long random input string the probability to remain all the time in $X$ is very small. And if it is smaller that $1/|X|$, there is an input string that works for all $q\in X$.}

So $X$ has no outgoing edges (and therefore is a strongly connected component of $S$'s graph). Now the same argument shows that there exists a string $u$ that forces $S$ to visit all vertices of $X$ when started from any vertex in $X$. This string $u$ appears with positive density in $\alpha$. So either the selected subsequence is finite (if $X$ has no accepting vertices) or the selected subsequence has positive density (since every occurrence of $u$ means that at least one term is selected when $S$ visits the accepting vertex). Lemma~\ref{lem:density} is proven.
\end{proof}

This finishes the proof of Theorem~\ref{th:agafonov}.
\end{proof}

\begin{remark}
In our exposition Agafonov's theorem is decomposed in three parts: (a)~positive density lemma, (b)~the obvious remark saying that applying a finite-state strategy in the game with a subsequence selected by a finite automata, we implement a finite-state strategy\footnote{This argument can also be adapted to the compression language: we can keep the selected subsequence in the compressed form and keep the rest uncompressed~\cite{becher-heiber,fct2017}. However, this creates some technical problems since in this way the compressed version is a pair of strings, and the argument is more natural with martingales.} ; (c) characterization of normal sequences in terms of martingales. All these parts were present already in~\cite{schnorr-stimm}, but we presented the proof for the reader's convenience and to make the decomposition clear.

The original proof of Agafonov (published in a hard-to-find volume and in Russian) was recently made available in~\cite{seiller-simonsen} (this reference was provided by an anonymous reviewer; thanks!). The proof is rather technical (both in the original version and in the embelisshed account), requiring several pages of estimates and computations after some preparations.
\end{remark}

Now we return to the first claim of Theorem~\ref{th:fsg}. We postponed the proof of the following result: \emph{no finite-state martingale wins against a normal sequence}. Now we are ready to prove it and even a slightly stronger statement~\cite{schnorr-stimm}:

\begin{theorem}[Schnorr and Stimm]\label{th:s-and-s}
Assume that $\alpha$ is a normal sequence and $m$ is a finite-state martingale. Then either the values of $m$ on the prefixes of $\alpha$ are constant, starting from some moment, or they decrease exponentially fast.
\end{theorem}

\begin{proof}
First we use the same argument as in the proof of Lemma~\ref{lem:density}. Look at the states of the finite-state martingale and consider the set $X$ of states that appear infinitely often. As we have seen, this set has no outgoing edges and every state has positive $\liminf$-density. Consider some state $i$ from $X$. It has two outgoing edges. Theorem~\ref{th:agafonov} guarantees that these two edges are used equally often (in the limit), since every subsequence selected by a finite automaton is normal (imagine that this state is used as the unique accepting state in the selection rule). Does the martingale makes a non-trivial bet in the state $i$? If it does not, for all $i$, i.e., if the martingale makes equal bets in all states from $X$, then the capital remains constant since the game stays in $X$ starting from some moment. If there is some $i\in X$ where the bets are not equal (say, $p$ and $1-p$ fractions of capital are used, where $p\ne 1/2$), then using each of the two edges once, we multiply the capital by $4p(1-p)$, and this number is less than $1$. So we get a factor less than $1$ for every $i\in X$ with non-equal bets, and get factor $1$ for all $i\in X$ that have equal bets, so the capital decreases exponentially fast (recall the every state in $X$ is visited with positive density).
\end{proof}

\subsection[Multi-account gales and a more general notion of finite-state measures]{Multi-account gales and a more general notion of finite-state measures}
\label{subsec:multi-account}

As we have said, the finite-state $s$-gales in our sense are called \emph{$1$-account} $s$-gales in~\cite{dai-fsd}. In the same paper a more general notion, called \emph{finite-state $k$-account $s$-gales}, is considered. They can be defined as non-negative linear combinations of $1$-account finite-state $s$-gales. The intuitive motivation is clear: the gambler splits her capital into $k$ different ``accounts'' and for each account uses the finite-state strategy (but never transfers money between the accounts).

Obviously, for the dimension of individual sequences (the case we are considering) this does not change anything: to win against the sequence, a multi-account strategy should contain a winning sub-strategy for some account.

From the viewpoint of gambling strategy the notion of finite-state martingale looks quite reasonable. However, if we consider just output distribution of random processes with finitely many states, there is a natural generalization. Assume that a finite set of states is given, and one of them is chosen as an initial state.  Assume also that for every state there are several outgoing edges, each has some transition probability, and for every state the sum of the probabilities for all outgoing edges equals $1$. This defines a random walk, and if we add for each edge a bit label ($0$ or $1$), we get a probability distribution on the infinite bit sequences. (The difference with the previous definition is that now the state is not determined by the output string.) The output distribution of this type can then be used to define martingales and $s$-gales in the same way, giving a more general definition.

It is easy to see that $k$-account gales become a special case of this definition (the splitting of the money between accounts is replaced by a probabilistic choice on the first step). Also one can define the version of $\KA_R$ based on this more general definition (taking the maximum over all states as initial states), and this version can be also used to characterize the finite-state dimension, as the following lemma implies.

\begin{lemma}\label{lem:general-ka}
The function $\KA_R$ in this general version is also a superadditive upper bound for a priori complexity. 
\end{lemma}

\begin{proof}
Indeed, the distribution obtained for every fixed initial state is a computable measure on the Cantor space, so for each initial state we get an upper bound for a priori complexity, and the minimum of these bounds is still an upper bound.

To show the superadditivity, we cannot anymore use the equality $\rho_i(uv)=\rho_i(u)\rho_j(v)$ since now the process can be in different states after output $u$. We need to replace $\rho_j(v)$ by a weighted sum of $\rho_k(v)$ for different $k$ that can be the states after output $u$. But since we take the maximum of $\rho_k(v)$ for all $k$ when defining $\KA_R$, we still have the superadditivity property.
\end{proof}

\nb{Probably there are some standard names? Markov chains?}

\subsection{Strong dimensions}\label{subsec:strong}

The notions of Hausdorff dimension and effective Hausdorff dimension have ``strong counterparts'': the notion of \emph{packing dimension} (called also modified upper box dimension) and \emph{constructive strong dimension}~\cite{athreya-strong}. Theorem 4.4 in~\cite{athreya-strong} says that the packing dimension of a set $X$ of binary sequences  equals the infimum of all $s$ such that there exist an $s$-gale that strongly wins (converges to infinity) on all elements of $X$. In the same paper~\cite{athreya-strong} it is shown that the strong dimension of a singleton $\{a_0a_1\ldots\}$ equals $\limsup_n \KS(a_0\ldots a_{n-1})/n$ (Corollary 6.2). This parallelism extends to the finite-state case: there is a natural notion of a \emph{strong finite-state dimension} of a sequence $\alpha=a_0a_1a_2\ldots$. Like the notion of finite-state dimension, it can be defined in many equivalent ways.

\begin{theorem}\label{th:strong}
The following definitions of the strong finite-state dimension are equivalent:

\begin{itemize}
\item the infimum of all $s$ such that there exists a finite-state $s$-gale that \emph{strongly} wins against $\alpha$, i.e., the capital converges to infinity when playing against $\alpha$.

\item the same, for $k$-account $s$-gales instead of $1$-account finite state $s$-gales.

\item for a given automatic relation $R$ consider $\limsup_n \KS_R(a_0\ldots a_{n-1})/n$; take the infimum over all $R$.

\item for a given finite-state probabilistic process $R$ consider $\limsup_n \KA_R(a_0\ldots a_{n-1})/n$; take the infimum over all $R$.

\item for a given superadditive function $F$ on strings that satisfies the calibration condition of Theorem~\ref{th:stateless1} consider $\limsup_n F(a_0\ldots a_{n-1})/n$; take infimum over all $F$.

\item for a given $k$ split $\alpha$ into disjoint $k$-bit blocks; then, for a given $N$, consider the distribution on the first $N$ of these $k$-bit blocks and its entropy $H^\mathrm{a}_{k,N}(\alpha)$. Let $\bar{H}^\mathrm{a}_k (\alpha)=\limsup_{N\to\infty} H^\mathrm{a}_{k,N}(\alpha)$; consider $\lim_k \bar{H}^\mathrm{a}_k (\alpha)/k$ or $\inf_k \bar{H}^\mathrm{a}_k (\alpha)/k$.

\end{itemize}

\end{theorem}

In this theorem some variations are allowed: (1)~one may consider more general $s$-gales based on the output distribution of finite-state probabilistic processes (see the discussion in Section~\ref{subsec:multi-account}); (2)~one may use the condition from Theorem~\ref{th:stateless} as a calibration condition; (3)~instead of aligned blocks, we may consider non-aligned ones.

\begin{proof}[Proof sketch]
All these equivalence statement are parallel to corresponding results about finite-state dimension. The difference is that we now consider strongly winning gales (not only unbounded but converging to infinity) and $\limsup$ instead of $\liminf$.

The equivalence proof go in the same way. The only new element is that now, knowing that $\limsup$ of block entropies is small, we cannot choose one limit distribution and adapt the coding (compression) to this specific distribution. Instead, we need a prefix code that works efficiently for all strings of some length (not losing much to a code that is adapted to a given distribution). Considering $k$-blocks as letters of some large alphabet $B$, we may use the following lemma.

\begin{lemma}[Universal coding lemma]
Let $B$ be some alphabet. For a given $N$ there exists a prefix-free code for $N$-bit strings over $B$ such that every string $z$ of length $N$ has a code of length $NH(z)+O(\log N)$, where $H(z)$ is the entropy of the distribution of letters in $z$, and the constant in $O(\cdot)$-notation depends on $B$ but not on $N$.
\end{lemma}

To prove the lemma, we can encode $z$ by first specifying the length of $z$ and number of occurrences of each letter in $z$. Then we may use Shannon--Fano code based on the letters' frequencies to specify $z$, because at that point decoder already knows the frequencies. In the language of Kolmogorov complexity, this lemma is a well-known bound $\KP(z)\le NH(z)+O(\log N)$ that was mentions already in Kolmogorov's papers from 1960s. This lemma also has an measure version saying that there exists a measure $\mu$ on all strings of length $N$ over $B$ such that $\mu(z)$ for every $z$ is at least $P_z(x)/\poly(N)$, where $P(z)$ is the Bernoulli distribution on $B$-strings based on the frequencies of letters in $z$. In the language of Kolmogorov complexity, we may use $\KA(z)$ instead of $\KP(z)$.
\end{proof}

Let us comment on the history of the notion of strong finite-state dimension of a bit sequence. Originally it appeared in the paper of Lempel and Ziv~\cite{ziv-lempel-variable-rate} where it was defined in terms of variable-rate block encoders with finite memory. This is close to the automatic complexity definition; however, definition in~\cite{ziv-lempel-variable-rate} did not introduce the notion of automatic complexity and used instead finite-state compressors such that decompression is unique if the initial and finite states are given in addition to the compressed output. Theorem~3 from their paper~\cite[p.~534]{ziv-lempel-variable-rate} says that this notion, denoted there by $\rho(\cdot)$, can be equivalently defined in terms of $\limsup$ of entropies of \emph{unaligned} blocks. Much later, in 2002 (the year when the arxiv version was published, the journal version appeared in 2007) Athreya, Hitchcock, Lutz and Mayordomo noted~\cite[Theorem 6.18]{athreya-strong} that the notion introduced by Ziv and Lempel (Athreya et al. refer to Ziv's paper~\cite{ziv}, not to~\cite{ziv-lempel-variable-rate}, but this is most probably a typo) coincides with the finite-state strong dimension defined in terms of gales. They also note that one can use both $1$-account and $k$-account gales in this equivalence proof (though they do not consider more general notion of output distribution of a finite-state probabilistic process). In a later paper (2005) Bourke, Hitchcock and Vinodchandran~\cite[Theorem 5.3]{bourke}  cite the result of Ziv and Lempel but strangely use aligned blocks (without explaining why aligned and non-aligned blocks give the same notion of strong finite-state dimension).  Theorem~\ref{th:strong}  includes all these results and shows that they may be obtained almost for free by using  the notions of automatic complexity, finite-state a priori probability and their superadditivity, while the original proofs were rather technical and were scattered among several papers.

\section{Discussion}\label{sec:discussion}

The connection between normality and finite-state computations was noticed long ago, as the title of~\cite{agafonov} shows; see also~\cite{schnorr-stimm} where normality was related to martingales arising from finite automata. This connection led to a characterization of normality as incompressibility (see~\cite{becher-heiber} for a direct proof). On the other hand, it was also clear that the notion of Kolmogorov complexity is not directly practical since it considers arbitrary algorithms as decompressors, and this makes it non-computable. So restricted classes of decompressors are of interest, and finite-state computations are a natural candidate for such a class.

Shallit and Wang~\cite{shallit-wang} suggested to consider, for a given string $x$, the minimal number of states in an automaton that accepts $x$ but not other strings of the same length. Later Hyde and Kjos-Hanssen~\cite{hkh} considered a similar notion using nondeterministic automata. The intrinsic problem of this approach is that it is not naturally ``calibrated'' in the following sense: measuring the information in bits, we would like to have about $2^n$ objects of complexity at most~$n$.  A ``calibrated'' approach was suggested by Calude, Salomaa and Roblot~\cite{csr}, see also~\cite{staiger-talk}; we have already discussed their definition in Section~\ref{subsec:csr}.

The incompressibility notion used in~\cite{becher-heiber} provides such a characterization for yet another approach to automatic complexity. It uses deterministic transducers applied to a sequence whose complexity is measured: transducers are used as compressors, not decompressors.  Still the decompression should be possible: Becher and Heiber require additionally that for every output string $y$ and every final state $s$ there is at most one input string that produces $y$ and brings the automaton into the state $s$. In~\cite{bch2} (see also~\cite{carton-talk}) Becher, Carton and Heiber consider a weaker condition: each output string has $O(1)$-preimages. The difference with our approach is that we do not consider the compression step at all, consider non-deterministic automata without initial/final states and require that decompressor is an $O(1)$-valued function. The proofs become simpler for two reasons: (1)~we compare the automatic complexity and Kolmogorov complexity and use standard results about Kolmogorov complexity; (2)~we explicitly state and prove the superadditivity property $\KS_R(xy)\ge \KS_R(x)+\KS_R(y)$ that is crucial for the proofs.

The reviewers of the previous version of this paper pointed out that Doty and Moser~\cite{doty-moser} were the first who characterized the finite-state dimension of a sequence in terms of automatic complexity based on decompressors. We overlooked this paper (from 2006) and apologize to its authors for not mentioning it  in~\cite{fct2017,fct2019}. They consider the finite-state transducers that have finite number of states, an initial state, transition and output functions. The transition function says what is the next state if the current state and input letter are known; the output function specifies the output string for each state and input letter. A transducer of this type determines a mapping from input strings to output strings. We assume that both input and output are binary strings.

Fix some way to measure the size of the tranducer. For example, we may define the size as the maximum of the number of states and the length of the strings that are values of the output function. For a given $k$, define the automatic complexity $\KS_k(x)$ of a string $x$ as the minimal length of an input string that is mapped to $x$ by some transducer of size at most $k$. Then the finite-state dimension of an infinite sequence $\alpha=a_0a_1\ldots$ is equal to $\inf_k \liminf_n \KS_k(a_0\ldots a_{n-1})/n$. Note that $\KS_k$ decreases as $k$ increases so $\inf_k$ can be replaced by $\lim_k$. This was proven by Doty and Moser~\cite{doty-moser} as well as the similar result for strong dimension and $\limsup$.  They used the previous results characterizing the finite state dimension in terms of decompressors~\cite{athreya-strong} and a rather complicated combinatorial construction converting an arbitrary transducer to a compressor with unique decompression. This needed this construction since the function $\KS_k$ is not superadditive. However, this construction can be easily avoided if we use our tools and note that for every transducer its input-output relation is contained in some automatic description mode and this can be extended to a finite set of transducers by Proposition~\ref{prop:modes}, (a). For the other direction, we need another obvious remark: a prefix-free decoding is performed by a transducer. However, this paper relates this notion to the measure of complexity introduced by Lempe

Another (and earlier) paper pointed out by reviewers is~\cite{sheinwald} where the complexity measure using decompressors is considered. It is shown there that the requirement that the compression is performed by an automaton is redundant in some sense: it is enough to require that decompression is performed by an automaton. This paper does this by relating the complexity measure defined in terms of decompressors to the complexity measure introduced in~\cite{lempel-ziv-complexity}; this makes the entire argument quite complicated.

An interesting open question is to find out the relations between different automatic complexity notions. Is there any formal relation between automatic complexity as defined in Section~\ref{sec:complexity} and notions of finite-state a priori complexity as defined in Section~\ref{subsec:fsm} and~\ref{subsec:multi-account}? Does the generalization of the class of probabilistic processes (Section~\ref{subsec:multi-account}) change the class of the corresponding complexity functions? Note that all three notions (automatic complexity and two finite-state a priori complexities) can be used to characterize normality since they are all superadditive and are upper bounds for Kolmogorov complexity with logarithmic precision. Still the results showing the different definitions of Kolmogorov complexity (using a priori probability and description modes) are close to each other, do not imply that the finite-state versions of the same notions are also close to each other.

\nb{May be, some connection can be established through the comparison with entropies? It would be probably rather weak, but may be still something reasonable\ldots}

\subsection*{Acknowledgments} We are grateful to Veronica Becher, Olivier  Carton and Pablo Heiber for many discussions of their paper~\cite{bch} and the relations between incompressibility and normality, and for the permission to use the observations made during these discussions in the current paper.

The initial version of this paper (that does not deal with finite-state dimension and does not contain the complexity criterion for normality) was presented at Fundamentals of Computation Theory  symposium in 2017~\cite{fct2017} and is available in \texttt{arxiv} as~\cite{arxiv2017}. The results about finite-state dimension and the complexity criterion were presented at Fundamentals of Computation Theory symposium in 2019~\cite{fct2019}.

We are grateful to the colleagues in LIRMM (ESCAPE team) and Moscow (Kolmogorov seminar, Computer Science Department of the HSE), and to the anonymous reviewers of STACS and FCT conferences and of JCSS who provided many important comments and pointed out several references that we overlooked. Many of their remarks are taken into account in the current version.



\begin{thebibliography}{9} 

\bibitem{agafonov} V.N.~Agafonov. Normal sequences and finite automata, \emph{Doklady AN SSSR}, \textbf{179}, 255--256 (1968).  
See also the paper of V.N.~Agafonov with the same name in the collection: \emph{Problemy Kibernetiki} (Cybernetics problems). Volume 20. Moscow: Nauka, 1968, p.~123--129.

\bibitem{alvarez-perfect}
N.~Alvarez, V.~Becher, P.~Ferrari, S.~Yuhjtman. Perfect necklaces, \emph{Advances of Applied Mathematics}, \textbf{80}, 48--61 (2016).

\bibitem{alvarez-independent}
N.~Alvarez, V.~Becher, O.~Carton, Finite-state independence and normal sequences, \emph{Journal of Computer and System Sciences}, \textbf{103}, 1--17 (2019).

\bibitem{athreya-strong} K.~Athreya, J.~Hitchcock, J.~Lutz, E.~Mayordomo, Effective strong dimension in algorithmic information and computational complexity, \emph{SIAM Journal on Computing}, \textbf{37}(3), 671--705 (2007), \url{http://epubs.siam.org/doi/10.1137/S0097539703446912}, see also~\url{http://arxiv.org/abs/cs/0211025}.

\nb{Packing dimension defined; its equivalence to modified upper box dimension is noted (without proof); gale characterization of packing dimension is provided (Theorem 4.4). The constructive strong dimension is defined in terms of lower semicomputable gales (or, equivalently, supergales) that \emph{strongly} win along a sequence (the capital is not only unbounded, but also converges to infinity). Corollary 6.2 says that constructive strong dimension of a sequence equals $\limsup \KS(x_1\ldots x_n)/n$.}

\bibitem{becher-turing}
V.~Becher, Turing's normal numbers: towards randomness, \emph{How the World Computes}, Proceedings of the Turing Centenary Conference and the 8th Conference in Computability in Europe, CiE2012, Cambridge, UK, June 18--23, 2012. Lecture notes in computer science, \emph{7318}, 35--45. Springer-Verlag, 2012.

\bibitem{bc}
V.~Becher, O.~Carton, Normal numbers and computer science, chapter in \emph{Sequences, groups and number theory}, V.~Berth\'e and M.~Rig\'o, editors, Trend in Mathematics Series, Birkhauser/Springer, 2018, 233--269, \url{https://link.springer.com/chapter/10.1007\%2F978-3-319-69152-7_7}.

\bibitem{bch} V.~Becher, O.~Carton, P.~Heiber,  Finite-state independence, \emph{Theory of Computing Systems}, \textbf{62}(7):1555--1272 (2018),
 See also \url{https://arxiv.org/pdf/1611.03921.pdf}.

\bibitem{bch2} V.~Becher, O.~Carton, P.~Heiber, Normality and automata, \emph{Journal of Computer and System Sciences}, \textbf{81}(8): 1592--1613 (2015).

\bibitem{becher-heiber} 
V.~Becher, P.~Heiber, Normal number and finite automata, \emph{Theoretical Computer Science}, \textbf{477}, 109--116 (2013).

\bibitem{berstel} J.~Berstel, \emph{Transductions and Context-Free Languages} Vieweg+Teubner Verlag, 1969.  ISBN: 978-3-519-02340-1. (For a revised 2006--2009 version see the author's homepage, \url{http://www-igm.univ-mlv.fr/~berstel}.)

\bibitem{besicovitch}
Abram S.~Besicovitch, The asymptotic distribution of the numerals in the decimal representation of the squares of the natural numbers, \emph{Mathematische Zeitschrift}, \textbf{39}(1), 146--156 (1935), \url{http://link.springer.com/10.1007/BF01201350}.

\bibitem{bss}
Laurent Bienvenu, Glenn Shafer, Alexander Shen, On the history of martingales in the study of randomness, \emph{Electronic Journ@l for History of Probability and Statistics}, \textbf{5}(1), June 2009, 1--40, \url{http://www.jehps.net/juin2009/BienvenuShaferShen.pdf}.

\bibitem{borel1909}
\'Emile Borel (Paris), Les probabilit\'es d\'enombrables et leur applications arithm\'etiques, \emph{Rendiconti del Circolo Matematico di Palermo}, \textbf{27}(1), 247--271 (1909), \url{https://doi.org/10.1007/BF03019651}.

\nb{Chapitre II, Les fractions d\'ecimales. Almost all reals have limit frequencies of digits (binary, decimal, etc.) Numbers with this property are called \emph{simplement normal}. A number $x$ is \emph{enti\`erement normal} w.r.t. base $10$ if $x, 10x,100x,\ldots$ are all simply normal with respect to all bases $10,100,1000$, etc. (aligned definition of normality). A number is \emph{absolument normal} if it enti\`erement normal w.r.t. all bases. An example of a normal number (similar to Champernowne, but with some complications that seem to be unnecessary), section 14.}

\bibitem{bourke}
Chris Bourke, John M.~Hitchcock, N.V.~Vinodchandran, Entropy rates and finite-state dimension, \emph{Theoretical Computer Science}, \textbf{349}(3), 392--406 (2005), \url{https://doi.org/10.1016/j.tcs.2005.09.040}.

\bibitem{bugeaud} Y.~Bugeaud, \emph{Distribution modulo one and Diophantine approximation}, Cambridge Tracts in Mathematics, 193, Cambridge University Press, 2012. 

\bibitem{csr}
C.S.~Calude, K.~Salomaa, T.K.~Roblot, Finite state complexity, \emph{Theoretical Computer Science}, \textbf{412}(41), 5668--5677 (2011).

\bibitem{css}
Cristian S. Calude, Ludwig Staiger, Frank Stephan, Finite state incompressible infinite sequences, \emph{Information and Computation}, \textbf{247}, 23--36 (2016), \url{https://doi.org/10.1016/j.ic.2015.11.003}.

\bibitem{carton-talk}
Olivier Carton, \emph{Normality and Automata}, talk at AutoMathA 2015, Leipzig, \url{http://www.automatha.uni-leipzig.de/slides/Carton.pdf}.

\bibitem{cassels1952}
J.W.S.~Cassels, On a paper of Niven and Zuckerman, \emph{Pacific Journal of Mathematics}, \textbf{2}(4), 555--557 (1952), \url{http://msp.org/pjm/1952/2-4/p06.xhtml}.

\bibitem{cassels1959}
J.W.S.~Cassels, On a problem of Steinhaus about normal numbers, \emph{Colloquium Mathematicum}, \textbf{VII}(1), 95--101 (1959), \url{http://matwbn.icm.edu.pl/ksiazki/cm/cm7/cm7120.pdf}.

\bibitem{champernowne}
D.~Champernowne, The construction of decimals normal in the scale of ten, \emph{Journal of the London Mathematical Society}, volume s1-8, issue 4 (October 1933; Received 19~April, read 27~April,1933), 254--260.

\bibitem{copeland-erdos}
Arthur H.~Copeland, Paul Erd\"os, Note on normal numbers, \emph{Bulletin of the American Mathematical Society}, \textbf{52}(10), 857--860 (1946), \url{https://projecteuclid.org/euclid.bams/1183509721}.

\bibitem{dai-fsd}
Jack J.~Dai, James I.~Lathrop, Jack H.~Lutz, Elvira Mayordomo, Finite-state dimension, \emph{Theoretical Computer Science}, \textbf{310}, issues 1--3, p.~1--33 (2004), \url{https://doi.org/10.1016/S0304-3975(03)00244-5}.

\nb{}

\bibitem{doty}
David Doty, Jack H.~Lutz, Satyadev Nandakumar, Finite-state dimension and real arithmetic, \emph{Information and Computation}, \textbf{205}(11), 1640--1651 (2007), \url{https://doi.org/10.1016/j.ic.2007.05.003}.

\bibitem{doty-moser}
David Doty, Philippe Moser, \emph{Finite-State Dimension and Lossy Decompressors}, arXiv:cs/0609096v2 preprint, \url{https://arxiv.org/pdf/cs/0609096.pdf}.

\bibitem{downey-hirschfeldt}
R.G.~Downey, D.R.~Hirschfeldt, \emph{Algorithmic randomness and complexity}, Springer, 2010, ISBN 978-0-387-68441-3, xxviii+855 p.

\bibitem{erdos-davenport} 
H.~Davenport, P.~Erd\"os, Note on Normal Decimals, \emph{Canadian Journal of Mathematics}, \textbf{4}, 58--63 (1952), \url{https://doi.org/10.4153/CJM-1952-005-3}.

\nb{Prove the conjecture from~\cite{copeland-erdos}: if $f(x)$ is a polynomial such that all $f(0),f(1),\ldots$ are positive integer, then the number $f(0)f(1)f(2)\ldots$ is normal.}

\bibitem{hkh} K.K.~Hyde, B.~Kjos-Hanssen, Nondeterministic complexity of overlap-free and almost square-free words, \emph{The Electronic Journal of Combinatorics}, \textbf{22}:3 (2015), \url{https://doi.org/10.37236/4851}.

\bibitem{fct2019}
Alexander Kozachinskiy, Alexander Shen, Two Characterizations for Finite-State Dimension, \emph{Fundamentals of Computation Theory}, 22nd International Symposium, FCT 2019 Copenhagen, Denmark, August 12--14, 2019, Springer, Lecture Notes in Computer Science, v.~11651 (2019), 80--94.

\bibitem{niederreiter} L.~Kuipers, H.~Niederreiter, \emph{Uniform distribution of sequences}.  John Wiley \& Sons, 1974.

\bibitem{lempel-ziv-complexity}
A.~Lempel, J.~Ziv, On the complexity of finite sequences, \emph{IEEE Transactions on Information theory}, \textbf{22}(1), 75--81 (January 1976), \url{https://ieeexplore.ieee.org/document/1055501}.

\nb{The first paper among 4 related papers. A production process for a string by adding its substrings (more or less) is described, the complexity is the minimal number of operations, and the greedy property is proven, i.e., there is an explicit way to construct this sequence that is optimal. The typical complexity (number of operations) is shown to be about $n/\log n$ for $n$-bit strings. No general definition of compression rates.}


\bibitem{lv}
M.~Li, P.~Vit\'anyi, \emph{An Introduction to Kolmogorov complexity and its applications}, 3rd ed., Springer, 2008 (1 ed., 1993; 2 ed., 1997), 792~pp. ISBN 978-0-387-49820-1.

\bibitem{lutz-dimension}
Jack H.~Lutz, Dimension in complexity classes. \emph{SIAM Journal on Computing}, \textbf{32}(5), 1236--1259, \url{https://doi.org/10.1137/S0097539701417723}. Preliminary version appeared in \emph{Proc. 15th IEEE Conference on Computational Complexity} (CCC 2000), p. 158-- 169.

\bibitem{lutz-individual}
Jack H.~Lutz, The dimension of individual strings and sequences, \emph{Information and Computation}, \textbf{187}(1), 49--79 (2003), \url{https://doi.org/10.1016/S0890-5401(03)00187-1}

\bibitem{mayordomo}
Elvira Mayordomo, A Kolmogorov complexity characterization of constructive Hausdorff dimension, \emph{Information Processing Letters}, \textbf{84}, 1--3 (2002), \url{https://doi.org/10.1016/S0020-0190(02)00343-5}.

\bibitem{maxfield1952}
John E.~Maxfield, A short proof of Pillai's theorem of normal numbers, \emph{Pacific Journal of Mathematics}, \textbf{2}(1), 23--24 (1952), \url{http://msp.org/pjm/1952/2-1/p04.xhtml}.

\nb{A proof of Pillai's theorem saying that aligned and non-aligned normality definitions are equivalent, based on the result of~\cite{niven-zuckerman1951}}

\bibitem{nies}
A.~Nies, \emph{Computability and randomness}, Oxford Logic Guides, Oxford University Press, 2009, ISBN~978-0199652600, 435 p.

\bibitem{niven1956}
Ivan Niven, \emph{Irrational Numbers}, The Mathematical Association of America, John Wiley and Sons, Inc., 1956, vii+164 p.  (Chapter VIII. Normal numbers, p.~94--116)

\nb{Normal numbers are defined in terms of non-aligned frequency. Describes the original definition of Borel that is equivalent to aligned normality of a sequence obtained by deleting an arbitrary finite prefix. ``S.S.Pillai established'' that one can just consider the aligned normality, S.S.Pillai, Proc. Ind. Acad. Sci, A, 12 (1940), 179--184. Non-normal numbers form a set of measure zero (Theorem~8.10). Non-aligned normality implies aligned normality (Theorem 8.13 and 8.14), and the uniform distribution of the orbit on $[0,1]$ (Theorem 8.15). Normality of Champernowne's number (Theorem 8.16). For the equivalence between aligned and non-aligned definitions gives references to~\cite{niven-zuckerman1951} and Cassels }

\bibitem{niven-zuckerman1951}
Ivan Niven and H.S.~Zuckerman, On the definition of normal numbers, \emph{Pacific Journal of Mathematics}, \textbf{1}(1), 103--109 (1951), \url{https://projecteuclid.org/euclid.pjm/1102613156}.

\nb{Note that the equivalence between non-aligned and aligned-with-shifts definition was mentioned without proof by Borel on p.~261 of~\cite{borel1909} and Hardy and Wright, provide a proof for it.}

\bibitem{piatetski-shapiro} 
I.I.~Piatetski-Shapiro, On the laws of distribution of the fractional parts of an exponential function [Russian:  \rus{И.\,И.\,Пятецкий-Шапиро, О законах распределения дробных долей показательной функции}],  Izvestia Akademii Nauk SSSR, Ser. Matem., \textbf{15}(1), 47--52 (1951). In Russian. \url{http://mi.mathnet.ru/izv3297}.

\bibitem{piatetski-shapiro1957} 
I.I.~Piatetski-Shapiro, On the distribution of the fractional parts of an exponential function [Russian: \rus{И.\,И.\,Пятецкий-Шапиро, О распределении дробных долей показательной функции}],  \emph{Uchenye Zapiski Moskovskogo Gosudarstvennogo Pedagogicheskogo Instituta}, \textbf{108}(2), 317--322 (1957). In Russian, \url{https://archive.org/details/piatetskishapiro1958}.

\nb{In \url{http://www.mathnet.ru/links/afcc78a0c8a7dfd74713cf87ba3e34fc/mzm206.pdf} Shkredov explains that in this paper a constant bound is replaced by some better one/ Shkredov has a wrong volume number , probably due to an incorrect interpretation of Roman numerals}

\bibitem{pillai1939}
S.S.~Pillai, On normal numbers,  \emph{Proceedings of the Indian Academy of Sciences}, Section A, \textbf{10}(1), 13--15 (1939), \url{https://www.ias.ac.in/article/fulltext/seca/010/01/0013-0015}.

\nb{(Math. Reviews MR0000020, review of H.S.~Zuckerman pointed out a problem with the proof: ``The author considers, among others, the number $123\cdots$ (in the scale $r$) formed by writing the positive integers (in the scale $r$) in succession. He gives a proof that these numbers are simply normal, that is, each digit from $0$ to $r-1$ appears with the asymptotic frequency $1/r$. The proof of the stronger statement that these numbers are normal is inadequate.'')}

\bibitem{pillai1940}
S.S.~Pillai, On normal numbers, \emph{Proceedings of the Indian Academy of Sciences}, Section A, \textbf{12}(2), 179--184 (August 1940), \url{https://www.ias.ac.in/article/fulltext/seca/012/02/0179-0184}.

\nb{Correction to~\cite{pillai1939}, proof of equivalence between aligned definition (A) and aligned definition with all shifts (B)}

\bibitem{postnikov1952}
A.G.~Postnikov, On the question of distribution of fractional parts of the exponential function [Russian: \rus{А.\,Г.\,Постников, К вопросу о распределении дробных долей показательной функции}], Doklady AN SSSR, \textbf{LXXXVI} (3), 473--476 (1952). In Russian.

\nb{No text available. The result is cited in \cite{piatetski-shapiro1957}.}

\bibitem{postnikov-piatetski1957}
A.G.~Postnikov, I.I.~Piatetski-Shapiro, Bernouilli-normal sequences of symbols [Russian: \rus{Нормальные по Бернулли последовательности знаков}], \emph{Izvestia AN SSSR, Seriya matematicheskaya}, \textbf{21}(4), 501--514 (1957), \url{http://mi.mathnet.ru/izv4030}. In Russian.

\nb{Normal sequences for arbitrary Bernoulli measures, strong law of large numbers as a consequence of ergodic theorem (!), Piatetski-Shapiro theorem from~\cite{piatetski-shapiro} proved for a constant factor and arbitrary Bernoulli distribution, without ergodic theory}

\bibitem{schmidt1960}
W.~Schmidt, On normal numbers, \emph{Pacific Journal of Mathematics}, \textbf{10}(2), 661--672 (1960), \url{https://msp.org/pjm/1960/10-2/pjm-v10-n2-p22-p.pdf}.

\nb{When the classes of normal reals w.r.t. two bases coincide}

\bibitem{schnorr1971}
C.~Schnorr, A unified approach to the definition of random sequences, \emph{Mathematical Systems Theory} (now \emph{Theory of Computing Systems}), \textbf{5}(3), 246--258~(1971), \url{https://doi.org/10.1007/BF01694181}.

\bibitem{schnorr-stimm}
C.~Schnorr, H.~Stimm, Endliche Automaten und Zufallsfolgen, \emph{Acta Informatica}, \textbf{1}(4), 345--39 (1972).

\bibitem{seiller-simonsen}
T.~Seiller, J.~Simonsen. \emph{An embellished account of Agafonov's proof of Agafonov's theorem} (2020). See \url{https://hal.archives-ouvertes.fr/hal-02891463/document} or \url{https://arxiv.org/abs/2007.03249}.

\bibitem{shallit-wang}
J.~Shallit, M.-W.~Wang, Automatic complexity of strings, \emph{Journal of Automata, Languages and Combinatorics}, \textbf{6}:4. 537-554 (April 2001)

\bibitem{sheinwald}
D.~Sheinwald, A.~Lempel, J.~Ziv, On encoding and decoding with two-way head machines, \emph{Information and computation}, \textbf{116}, 128--133 (1995), url{https://www.sciencedirect.com/science/article/pii/S0890540185710097}. Preliminary version: Data Compression Conferences, 1991, see \url{https://ieeexplore.ieee.org/document/213359}

\bibitem{shen-intro} 
A.~Shen, Around Kolmogorov complexity: basic notions and results. 
\emph{Measures of Complexity. Festschrift for Alexey Chervonenkis}, edited by V.~Vovk, H.~Papadoupoulos, A.~Gammerman, Springer, 2015,  p.~75--116, see also \url{http://arxiv.org/abs/1504.04955}.

\bibitem{fct2017}
A.~Shen, Automatic Kolmogorov Complexity and Normality Revisited, \emph{Fundamentals of Computation Theory, 2017, Proceedings}, Lecture Notes in Computer Science, vol.~10472, 418--430.

\bibitem{arxiv2017}
A.~Shen, \emph{Automatic Kolmogorov Complexity and Normality Revisited}, \url{https://arxiv.org/abs/1701.09060v1}, first partial version of this paper.

\bibitem{kolmbook} 
A.~Shen, V.A.~Uspensky, N.~Vereshchagin, \emph{Kolmogorov complexity and algorithmic randomness}, Moscow, MCCME, 2013 (In Russian). English version published by AMS, see~\url{http://www.ams.org/publications/authors/books/postpub/surv-220} and \url{www.lirmm.fr/~ashen/kolmbook-eng-scan.pdf}.

\bibitem{staiger2005}
L.~Staiger, Constructive dimension equals Kolmogorov complexity, \emph{Information Processing Letters}, \textbf{93}(3), 149--153 (2005), \url{https://doi.org/10.1016/j.ipl.2004.09.023}.

\bibitem{staiger-talk}
L.~Staiger, \emph{Finite Automata and Randomness}, talk at AutoMathA 2015, Leipzig, \url{http://www.automatha.uni-leipzig.de/slides/Staiger.pdf}.

\bibitem{uspshen}
V.A.~Uspensky, A.~Shen, Relations between varieties of Kolmogorov complexities, \emph{Mathematical Systems Theory}, \textbf{29}, 271--292 (1996).

\bibitem{ville1939} J.~Ville, \emph{\'Etude critique de la notion de collectif}, Gauthier-Villars, Paris, 1939, \url{http://www.numdam.org/issue/THESE_1939__218__1_0.pdf}. In French.

\bibitem{wall} D.D.~Wall, \emph{Normal numbers}, Thesis, University of California, 1949.

\nb{No text available}

\bibitem{weber} A.~Weber, On the valuedness of finite transducers, \emph{Acta Informatica}, \textbf{27}(8), 749--780 (1990).

\bibitem{weyl1916}
 H.~Weyl
, \"Uber dir Gleichverteilung von Zahlen mod. Eins, \emph{Mathematische Annalen}, \textbf{77}(3), 313--352 (1916), \url{http://link.springer.com/10.1007/BF01475864}.
 
 \nb{Uniformly distributed points on a circle: definition, criterion with integration of continuous function, Fourier series, uniformity of $\{n\alpha\}$ for irrational $\alpha$, several dimensions. Satz 13: if $\xi$ is irrational, then $n^2\xi \bmod 1$ are uniformly distributed}

\bibitem{ziv} 
J.~Ziv, Coding theorems for individual sequences, \emph{IEEE Transactions on Information Theory}, \textbf{24}(4), 405--412 (July 1978), \url{https://ieeexplore.ieee.org/document/1055911}

\nb{Paper \#3 in the series. An encoding by a finite-state automata with constant rate (blocks of letter of fixed side are transformed to letters of a large alphabet by a finite-state automaton. In this setting we need to be ready to encode any string that appeas in the input sequence, therefore the topological entropy of the sequence gives a lower bound for compression rate. Then an approximate setting is considered where small fraction of errors after decoding is allowed. In this way we need to consider, for every $\varepsilon>0$, the minimal topological entropy of sequences that are $\varepsilon$-Besicovitch-close ot a given sequence. It is shown that some encoding reaches the bound provided by this quantity.}

\bibitem{ziv-lempel-universal} 
J.~Ziv, A.~Lempel, A universal algorithm for sequential data compression, \emph{IEEE Transactions on Information Theory}, \textbf{23}(3), 337--343 (May 1977), \url{https://ieeexplore.ieee.org/document/1055714}

\nb{Paper \#2 in the series. A specific encoding algorithm that is a ``variation'' of the construction of paper \#1 is considered. A notion of \emph{source} is introduced that is similar to the notion of a subshift, and $h(L)$ is defined as the logarithm of a number of strings of length $L$ that can appear, divided by $L$, like in the definition of the topological entropy. ``In this paper we investigate the performance of the proposed compression algorithm with respect to a non-probabilistic model of constrained information sources'' (p.~339); it is compared to lower bound valid for all block codes.  No analysis for the average case and no references to Shannon entropy. Wikipedia refers to this page speaking about LZ77 compression algorithm}

\bibitem{ziv-lempel-variable-rate}
J.~Ziv, A.~Lempel, Compression of Individual Sequences via Variable-Rate Coding, \emph{IEEE Transactions on Information Theory}, \textbf{24}(5), 530--536 (September 1978), \url{https://ieeexplore.ieee.org/document/1055934}

\nb{Paper \#4 in the serier. Wikipedia refers to it when speaking about LZ78 algorithm. This paper considers variable-rate encodings with finite memory. Fix some sequence $x$ that should be encoded. For every $s$, and for every $n$ the best $s$-state encoder for $n$-bit prefix of a $x$ is considered, the compression ratio is measured, and then (for fixed $s$) the $\limsup$ is considered as the length goes to infinity. Finally, we let $s\to\infty$ and take the limit $\rho(x)$. It is shown (Theorem 3, page 534) that $\rho(x)$ is equal to the strong finite-state entropy defined in terms of \emph{nonaligned} (see equation (19), p.~534) blocks. Also another quantity is considered (in theorems 1 and 2): for every string $x$ they define $c(x)$ as the maximal number of different blocks for all splitting of $x$ into blocks (of arbitrary size). They say that LZ78 algorithm is somehow optimal in terms of this quantity, and even claim some connection with $\rho(x)$, but the connection remains unclear. It would be interesting to find out whether some characterization of finite-state strong dimension (i.e., $\rho(x)$) can be given in terms of $c$-values for prefixes.}

\end{thebibliography}
\end{document}